\documentclass[a4paper]{article}

\usepackage{graphicx, amssymb, latexsym, amsfonts, amsmath, lscape, amscd, subcaption, todonotes,
amsthm, color, epsfig, mathrsfs, tikz, enumerate,soul}
\usepackage{tabularx}
\usetikzlibrary{arrows.meta}
\usepackage{hyperref}
\newtheorem{theorem}{Theorem}[section]

\newtheorem{corollary}[theorem]{Corollary}

\newtheorem{lemma}[theorem]{Lemma}
\newtheorem{proposition}[theorem]{Proposition}

\newtheorem{problem}[theorem]{Problem}
\newtheorem{observation}[theorem]{Observation}

\newtheorem{rem}[theorem]{Remark}

\newcommand\DELETE[1]{}

\hypersetup{
	colorlinks=true, 
	linkcolor=red,
	citecolor=blue,
}
\usepackage{todonotes}

\begin{document}


\title{On arc-density of pushably $3$-critical oriented graphs}

\author{Tapas Das\thanks{Indian Institute of Technology Dharwad, India. Email: {tapasdas625@gmail.com}. } ~~ Pavan P D\thanks{University of Turku, FI-20014, Turku, Finland. Supported by Research Council of Finland grants 338797 and 358718. Email: {pavanpdevaraj@gmail.com}} ~~ Sagnik Sen\thanks{Indian Institute of Technology Dharwad, India. Email: {sen007isi@gmail.com}. } ~~
S Taruni\thanks{Centro de Modelamiento Matemático (CNRS IRL2807), Universidad de Chile, Santiago, Chile. Supported by Centro de Modelamiento Matemático (CMM) BASAL fund FB210005 for center of excellence from ANID-Chile. Email: {tsridhar@cmm.uchile.cl}. } ~~ }

\date{}

\maketitle

\begin{abstract}
An oriented graph $\overrightarrow{G}$ is pushably $k$-critical if it is not pushably $k$-colorable, but every proper subgraph of 
$\overrightarrow{G}$ is. The main result of this article is that every pushably $3$-critical oriented graph on $n$ vertices, but for four exceptions, has at least $\frac{15n+2}{13}$ arcs, and that this bound is tight. As an application of this result, we show that the class of oriented graphs with maximum average degree strictly less than $\frac{30}{13}$ and girth at least $5$, which includes all oriented planar and projective planar  graphs with girth at least $15$, have pushable chromatic number at most $3$. Moreover, we provide an exhaustive list of pushably $3$-critical graphs with maximum average degree equal to $\frac{30}{13}$ and a pushably $3$-critical orientation of a $4$-cycle to prove the tightness of our bound with respect to both maximum average degree and girth. 
We also show that these classes of oriented graphs admit a homomorphism to an oriented planar graph on six vertices (an orientation of $K_{2,2,2}$)
which (tightly) improves a result due to Borodin \textit{et al.} [Discrete Mathematics 1998]. 
Furthermore, for these classes of oriented graphs, we prove that the 
$2$-dipath $L(p,q)$ and the oriented $L(p,q)$ spans are upper bounded by $2p+3q$ for all $q \leq p$. All these implications improve previously known results. 
\end{abstract}

\section{Introduction}
An \textit{oriented graph} is a directed graph without any loops or bidirectional arcs. In this article, we assume that all oriented graphs have a simple graph as their underlying graphs, unless otherwise stated. Given an oriented graph $\overrightarrow{G}$, $V(\overrightarrow{G})$ denotes the set of vertices, $A(\overrightarrow{G})$ denotes the set of arcs, and $G$ denotes its underlying simple graph.

A \textit{homomorphism} $f$ of an oriented graph $\overrightarrow{G}$ to another oriented graph $\overrightarrow{H}$ is a vertex mapping $f: V(\overrightarrow{G}) \to V(\overrightarrow{H})$ such that if $uv$ is an arc of $\overrightarrow{G}$, then $f(u)f(v)$ is an arc of $\overrightarrow{H}$. If $\overrightarrow{G}$ admits a homomorphism to $\overrightarrow{H}$, then we say that $\overrightarrow{G}$ is \textit{$\overrightarrow{H}$-colorable}, and denote it by $\overrightarrow{G} \to \overrightarrow{H}$. The vertices of $\overrightarrow{H}$ are called \textit{colors} in this context. 
In particular, if an oriented graph $\overrightarrow{G}$ is 
$\overrightarrow{H}$-colorable for some $\overrightarrow{H}$ on $k$ vertices via a homomorphism $f$ of $\overrightarrow{G}$ to $\overrightarrow{H}$, then
$\overrightarrow{G}$ is said to be \textit{oriented $k$-colorable}, and 
the function $f$ is called its \textit{oriented $k$-coloring}. 
The \textit{oriented chromatic number} of $\overrightarrow{G}$, 
denoted by $\chi_o(\overrightarrow{G})$, is the minimum $k$ such that $\overrightarrow{G}$ is oriented $k$-colorable.

To \textit{push} a vertex $v$ of an oriented graph $\overrightarrow{G}$ is to reverse the direction of the arcs incident to $v$. If instead of a vertex, a set of vertices $S \subseteq V(\overrightarrow{G})$ is pushed, then the so-obtained modified oriented graph is called a \textit{push equivalent} oriented graph $\overrightarrow{G}^S$ of $\overrightarrow{G}$, and the equivalence relation is expressed as $\overrightarrow{G} \equiv_p \overrightarrow{G}^S$. 
Given an oriented graph $\overrightarrow{G}$, if it is possible to modify it 
by pushing a vertex subset $S \subseteq V(\overrightarrow{G})$ in such a way that 
there exists a homomorphism $f$ of $\overrightarrow{G}^S$ to $\overrightarrow{H}$, then we say that $f$ is a pushable homomorphism of $\overrightarrow{G}$ to 
$\overrightarrow{H}$. Moreover, we call $\overrightarrow{G}$ as 
\textit{pushably $\overrightarrow{H}$-colorable}, and denote it by $\overrightarrow{G} \xrightarrow{push} \overrightarrow{H}$.  In particular, if an oriented graph $\overrightarrow{G}$ is 
pushably $\overrightarrow{H}$-colorable for some $\overrightarrow{H}$ on $k$ vertices via a pushable homomorphism $f$ of 
$\overrightarrow{G}$ to $\overrightarrow{H}$, then
$\overrightarrow{G}$ is said to be \textit{pushably $k$-colorable}, and 
the function $f$ is called its \textit{pushable $k$-coloring}. 
The \textit{pushable chromatic number} of $\overrightarrow{G}$, 
denoted by $\chi_p(\overrightarrow{G})$, is the minimum $k$ such that 
$\overrightarrow{G}$ is pushably $k$-colorable. An oriented graph $\overrightarrow{G}$ is \textit{oriented $\overrightarrow{H}$-critical} (resp., \textit{oriented $k$-critical}, 
\textit{pushably $\overrightarrow{H}$-critical},
\textit{pushably $k$-critical})
if $\overrightarrow{G}$ is not 
$\overrightarrow{H}$-colorable 
(resp., oriented $k$-colorable,
pushably $\overrightarrow{H}$-colorable, 
pushably $k$-colorable) 
but every proper subgraph of 
$\overrightarrow{G}$ is.  Our main result is the following:

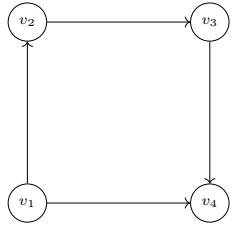
\begin{figure}[]
			
				\centering
                      
				\scalebox{.8}{
					
					\begin{tikzpicture}
						
						\node[draw, circle, minimum size=.05cm] (s1) at (0,0){\scriptsize$v_1$};
						\node[draw, circle, minimum size=.05cm] (s2) at (3,0) {\scriptsize$v_4$};
						\node[draw, circle, minimum size=.05cm] (s3) at  (3,3) {\scriptsize$v_3$};
                        \node[draw, circle, minimum size=.05cm] (s4) at  (0,3) {\scriptsize$v_2$};  
						
						\draw[->] (s1) -- (s2);
						\draw[->] (s3) -- (s2);
						\draw[<-] (s3) -- (s4);
                        \draw[->] (s1) -- (s4);
                   	
				\end{tikzpicture}}
\caption{The oriented graph $\overrightarrow{C}_{-4}$, a particular orientation of the $4$-cycle.}
\label{fig:clique c_-4}
		\end{figure}

        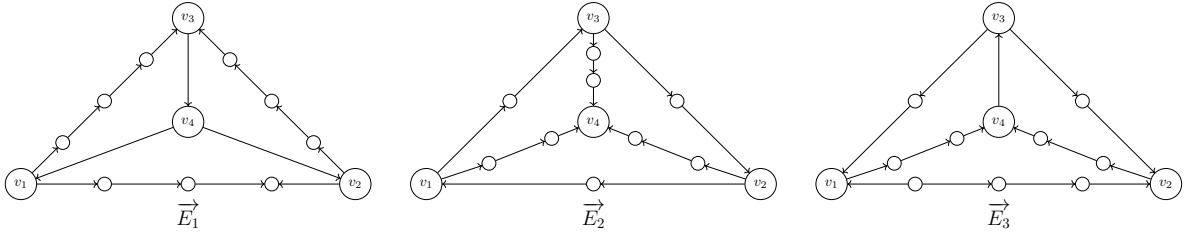
\begin{figure}[t]
    \begin{tabularx}{\textwidth}{XXX}
    
    {\centering
    \scalebox{.55}{
        \begin{tikzpicture}
        \node[draw, circle, minimum size=.75cm] (v1) at (0,0) {$v_1$};
        \node[draw, circle, minimum size=.75cm] (v2) at (8,0) {$v_2$};
        \node[draw, circle, minimum size=.75cm] (v3) at (4,4) {$v_3$};
        \node[draw, circle, minimum size=.75cm] (v4) at (4,1.5) {$v_4$};
        \node[draw, circle, minimum size=.05cm] (v5) at (1,1) {};
        \node[draw, circle, minimum size=.05cm] (v6) at (2,2) {};
        \node[draw, circle, minimum size=.05cm] (v7) at (3,3) {};
        \node[draw, circle, minimum size=.05cm] (v8) at (5,3) {};
        \node[draw, circle, minimum size=.05cm] (v9) at (6,2) {};
        \node[draw, circle, minimum size=.05cm] (v10) at (7,1) {};
        \node[draw, circle, minimum size=.05cm] (v11) at (2,0) {};
        \node[draw, circle, minimum size=.05cm] (v12) at (4,0) {};
        \node[draw, circle, minimum size=.05cm] (v13) at (6,0) {};
        \node at (4,-.75){\Large$\overrightarrow{E_1}$};
        \draw[->] (v1) -- (v5);
        \draw[->] (v5) -- (v6);
        \draw[->] (v6) -- (v7);
        \draw[->] (v7) -- (v3);
        \draw[<-] (v3) -- (v8);
        \draw[<-] (v8) -- (v9);
        \draw[<-] (v9) -- (v10);
        \draw[<-] (v10) -- (v2);
        \draw[<-] (v11) -- (v1);
        \draw[<-] (v12) -- (v11);
        \draw[<-] (v13) -- (v12);
        \draw[->] (v2) -- (v13);
        \draw[->] (v4) -- (v1);
        \draw[->] (v3) -- (v4);
        \draw[->] (v4) -- (v2);
        \end{tikzpicture}
    }}



    &

    {\centering
    \scalebox{.55}{
        \begin{tikzpicture}
        \node[draw, circle, minimum size=.75cm] (v1) at (0,0) {$v_1$};
        \node[draw, circle, minimum size=.75cm] (v2) at (8,0) {$v_2$};
        \node[draw, circle, minimum size=.75cm] (v3) at (4,4) {$v_3$};
        \node[draw, circle, minimum size=.75cm] (v4) at (4,1.5) {$v_4$};
        \node[draw, circle, minimum size=.05cm] (v14) at (4,3.15) {};
        \node[draw, circle, minimum size=.05cm] (v15) at (4,2.5) {};
        \node[draw, circle, minimum size=.05cm] (v7) at (4,0) {};
        \node[draw, circle, minimum size=.05cm] (v5) at (2,2) {};
        \node[draw, circle, minimum size=.05cm] (v6) at (6,2) {};
        \node[draw, circle, minimum size=.05cm] (v8) at (1.5,0.5) {};
        \node[draw, circle, minimum size=.05cm] (v9) at (3,1.1) {};
        \node[draw, circle, minimum size=.05cm] (v10) at (5,1.1) {};
        \node[draw, circle, minimum size=.05cm] (v11) at (6.5,.5) {};
        \node at (4,-.75){\Large$\overrightarrow{E_2}$};
        \draw[->] (v1) -- (v5);
        \draw[->] (v5) -- (v3);
        \draw[->] (v3) -- (v6);
        \draw[->] (v6) -- (v2);
        \draw[->] (v2) -- (v7);
        \draw[->] (v7) -- (v1);
        \draw[->] (v1) -- (v8);
        \draw[->] (v8) -- (v9);
        \draw[->] (v9) -- (v4);
        \draw[->] (v2) -- (v11);
        \draw[<-] (v10) -- (v11);
        \draw[->] (v10) -- (v4);
        \draw[<-] (v4) -- (v15);
        \draw[<-] (v15) -- (v14);
        \draw[<-] (v14) -- (v3);
        \end{tikzpicture}
    }}



    &

    {\centering
    \scalebox{.55}{
        \begin{tikzpicture}
        \node[draw, circle, minimum size=.75cm] (v1) at (0,0) {$v_1$};
        \node[draw, circle, minimum size=.75cm] (v2) at (8,0) {$v_2$};
        \node[draw, circle, minimum size=.75cm] (v3) at (4,4) {$v_3$};
        \node[draw, circle, minimum size=.75cm] (v4) at (4,1.5) {$v_4$};
        \node[draw, circle, minimum size=.05cm] (v5) at (2,2) {};
        \node[draw, circle, minimum size=.05cm] (v6) at (6,2) {};
        \node[draw, circle, minimum size=.05cm] (v7) at (4,0) {};
        \node[draw, circle, minimum size=.05cm] (v12) at (2,0) {};
        \node[draw, circle, minimum size=.05cm] (v13) at (6,0) {};
        \node[draw, circle, minimum size=.05cm] (v8) at (1.5,0.5) {};
        \node[draw, circle, minimum size=.05cm] (v9) at (3,1.1) {};
        \node[draw, circle, minimum size=.05cm] (v10) at (5,1.1) {};
        \node[draw, circle, minimum size=.05cm] (v11) at (6.5,.5) {};
        \node at (4,-.75){\Large$\overrightarrow{E_3}$};
        \draw[->] (v5) -- (v1);
        \draw[->] (v3) -- (v5);
        \draw[->] (v3) -- (v6);
        \draw[->] (v6) -- (v2);
        \draw[->] (v13) -- (v2);
        \draw[->] (v7) -- (v13);
        \draw[->] (v12) -- (v7);
        \draw[->] (v12) -- (v1);
        \draw[->] (v1) -- (v8);
        \draw[->] (v8) -- (v9);
        \draw[->] (v9) -- (v4);
        \draw[->] (v2) -- (v11);
        \draw[<-] (v10) -- (v11);
        \draw[->] (v10) -- (v4);
        \draw[->] (v4) -- (v3);
        \end{tikzpicture}
    }}



    \end{tabularx}
    \caption{An exhaustive (up to push equivalence) list  of pushably $3$-critical oriented graphs with $13$ vertices and $15$ arcs and maximum average degree equal to $\frac{30}{13}$.}
    \label{fig:Ei}
\end{figure}

\begin{theorem}\label{th push 3-critical density}
    Let $\overrightarrow{G}$ be a pushably $3$-critical oriented graph. 
    If $G \not\equiv_p Z$ where $Z \in \{\overrightarrow{C}_{-4},$ $\overrightarrow{E}_1,$ 
    $\overrightarrow{E}_2,$ $\overrightarrow{E}_3 \}$  (see Fig.~\ref{fig:clique c_-4} and Fig.~\ref{fig:Ei}), then 
    $$|A(\overrightarrow{G})| \geq \frac{15|V(\overrightarrow{G})| + 2}{13}.$$
    Moreover, this bound is tight. 
\end{theorem}

\subsection{Preliminaries}
Given a vertex $u$ of an oriented graph $\overrightarrow{G}$, let $N^-(u) = \{ v \in V(\overrightarrow{G}) : vu \in A(\overrightarrow{G}) \}$ and $N^+(u) = \{ v \in V(\overrightarrow{G}) : uv \in A(\overrightarrow{G}) \}$ denote its \textit{in-neighborhood} and \textit{out-neighborhood}, respectively. The \emph{in-degree} and \emph{out-degree} of $v$ are $deg^{-}(v) = |N^-(v)|$ and $deg^{+}(v) = |N^+(v)|$, respectively. The \emph{degree} $deg(v)$ of $v$ is given by $deg(v) = deg^{-}(v) + deg^{+}(v)$.  For a subset $X$ of vertices in $\overrightarrow{G}$, $\overrightarrow{G}[X]$ denotes the subgraph of $\overrightarrow{G}$ induced by $X$. For any two disjoint vertex subsets $X,Y \subseteq V(\overrightarrow{G})$, $[X,Y]$ and $A[X,Y]$ denotes the set and number of of edges (arcs) between $X$ and $Y$ in $\overrightarrow{G}$, respectively. Given an oriented graph $\overrightarrow{G}$ with the set of vertices $V(\overrightarrow{G}) = \{v_1, v_2, \ldots, v_n\}$,
its \textit{anti-twinned oriented graph} 
$AT(\overrightarrow{G})$ has the set of vertices 
$V(AT(\overrightarrow{G}))=\{v_1, v_2, \ldots, v_n\} \cup \{v'_1, v'_2, \ldots, v'_n\}$
and the set of arcs 
$A(AT(\overrightarrow{G}))=\{v_iv_j, v'_iv'_j, v'_jv_i, v_jv'_i | v_iv_j \in A(\overrightarrow{G})\}.$ 
That is, $AT(\overrightarrow{G})$ is obtained by taking two copies of $\overrightarrow{G}$, then adding arcs in between them in such a way that the vertices corresponding to each other in the two copies become twins (have the same in-neighborhood and out-neighborhood), and then pushing all vertices of the second copy of $\overrightarrow{G}$.  
It is important to note that homomorphisms and pushable homomorphisms are 
closely related through the anti-twinned oriented graph construction. 

\begin{proposition}[\cite{push}]\label{prop: basic KM}
Let $\overrightarrow{G}$ and $\overrightarrow{H}$ be two oriented graphs. Then the following holds.
\begin{enumerate}[(i)]
    \item We have $\overrightarrow{G} \xrightarrow{push} \overrightarrow{H}$ if and only if $\overrightarrow{G} \rightarrow AT(\overrightarrow{H})$. 
    \item If $\overrightarrow{G} \xrightarrow{push} \overrightarrow{H}$, 
    then for any $\overrightarrow{H}' \equiv_{p} \overrightarrow{H}$ there exists a $\overrightarrow{G}' \equiv_p \overrightarrow{G}$ satisfying 
    $\overrightarrow{G}' \rightarrow \overrightarrow{H}'$. 
    \item We have $\chi_p(\overrightarrow{G}) \leq \chi_o(\overrightarrow{G}) \leq 2 \chi_p(\overrightarrow{G})$. 
\end{enumerate} 
\end{proposition}

According to the definition of pushable chromatic number, notice that an oriented graph $\overrightarrow{G}$ 
admits a pushable $k$-coloring if and only if 
$\overrightarrow{G}$ admits a pushable homomorphism to some orientation of 
$K_k$. Since the two distinct orientations of $K_3$ (that is, the directed $3$-cycle and the transitive $3$-cycle) are push equivalent, 
Proposition~\ref{prop: basic KM}(i) and (ii) imply the following observations. 

 \begin{figure}[]
			\begin{tabularx}{\textwidth}{ X X}
				\centering
                      
				\scalebox{.6}{
					
					\begin{tikzpicture}
						
						\node[draw, circle, minimum size=.05cm] (s1) at (-2.75,-1.25){$\bar{2}$};
						\node[draw, circle, minimum size=.05cm] (s2) at (4.75,-1.25) {$\bar{1}$};
						\node[draw, circle, minimum size=.05cm] (s3) at  (1,3.5) {$\bar{0}$};

						\node at (1,-1.75){\Large$\overrightarrow{C_3}$};
						
						\draw[<-] (s1) -- (s2);
						\draw[->] (s3) -- (s2);
						\draw[<-] (s3) -- (s1);

				\end{tikzpicture}}
				
				\label{E1}
				
				&
				
				\scalebox{0.6}{
						
		\begin{tikzpicture}
			
			\node[draw, circle, minimum size=.75cm] (s1) at (2,1) {$\bar{2'}$};
			\node[draw, circle, minimum size=.75cm] (s2) at (0,1){$\bar{1'}$};
			
			\node[draw, circle, minimum size=.75cm] (s3) at  (1,-.25) {$\bar{0'}$};
			
			\node[draw, circle, minimum size=.75cm] (s4) at (-2.75,-1.25) {$\bar{2}$};
			
			\node[draw, circle, minimum size=.75cm] (s5) at (4.75,-1.25){$\bar{1}$};
			
			\node[draw, circle, minimum size=.75cm] (s6) at  (1,3.5) {$\bar{0}$};

            \node at (1,-1.75){\Large$AT(\overrightarrow{C_3})$};
			
			\draw[<-] (s1) -- (s2);
			\draw[<-] (s2) -- (s3);
			\draw[<-] (s3) -- (s1);
			\draw[<-] (s4) -- (s5);
			\draw[<-] (s5) -- (s6);
			\draw[<-] (s6) -- (s4);
			\draw[->] (s6) -- (s1);
			\draw[->] (s3) -- (s4);
			\draw[->] (s1) -- (s5);
			\draw[->] (s4) -- (s2);
			\draw[->] (s2) -- (s6);
			\draw[->] (s5) -- (s3);

	\end{tikzpicture}}
				
				\label{E1}

			\end{tabularx}
\caption{The directed $3$-cycle $\overrightarrow{C}_3$ and its anti-twinned oriented graph $AT(\overrightarrow{C}_3)$.}
\label{anti}
		\end{figure}
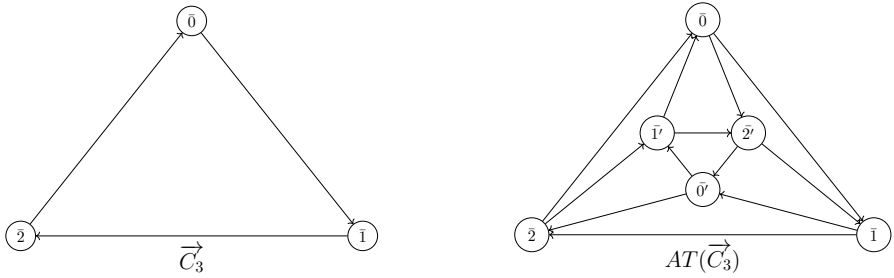

\begin{observation}\label{obs C3=3}
    Let $\overrightarrow{G}$ be an oriented graph, $\overrightarrow{C}_3$ denote the directed $3$-cycle, and $AT(\overrightarrow{C}_3)$ denote the anti-twinned oriented graph of $\overrightarrow{C}_3$ 
    (see Fig.~\ref{anti}). Then the following are equivalent. 
    \begin{enumerate}[(i)]
        \item $\overrightarrow{G}$ is pushably $3$-colorable,
        \item $\overrightarrow{G}$ is pushably $\overrightarrow{C}_3$-colorable,
        \item $\overrightarrow{G}$ is $AT(\overrightarrow{C}_3)$-colorable.
    \end{enumerate}
\end{observation}

From the complexity point of view, Klostermeyer and MacGillivray~\cite{push} 
characterized the complete dichotomy of the decision version of the problem of determining the pushable $\overrightarrow{H}$-colorability of an oriented graph. We summarize their result as follows. 

\begin{theorem}[\cite{push}]\label{th KM complexity}
    Let $\overrightarrow{H}$ be an oriented graph and $\overrightarrow{C}_4$ denote the directed $4$-cycle. Then the following holds. 
    \begin{enumerate}[(i)]
        \item Given an input oriented graph $\overrightarrow{G}$, 
        determining whether $\overrightarrow{G}$ is pushably 
        $\overrightarrow{H}$-colorable is NP-complete
        if $\overrightarrow{H}$ does not admit a homomorphism 
        to $\overrightarrow{C}_4$, and 
        is polynomial time solvable otherwise. 
                
        \item Given an input oriented graph $\overrightarrow{G}$, 
        determining whether $\overrightarrow{G}$ is pushably 
        $k$-colorable is NP-complete for $k \geq 3$, and is polynomial time solvable otherwise. 
    \end{enumerate}
\end{theorem}

\subsection{Motivations and applications}
A \textit{homomorphism} of a graph $G$ to another graph $H$ is an edge preserving vertex mapping, and $G$ is \textit{$k$-colorable} if 
$G$ admits a homomorphism $f$ to $K_k$, where $f$ is called a $k$-coloring of $G$. We say that $G$ is $k$-critical\footnote{In many papers~\cite{kostochka2014ore,kostochkagrotzch}, $G$ is called $(k+1)$-critical instead of $k$-critical.} if $G$ 
is not $k$-colorable, but every proper subgraph of $G$ is.

Theorem~\ref{th KM complexity}(i) is an analogue of the famous 
Hell-Ne\v{s}et\v{r}il Theorem~\cite{hell1990complexity} which states that given an input graph $G$, determining whether $G$ admits a homomorphism to $H$ is NP-complete if 
$H$ is not bipartite, and is polynomial time solvable otherwise. 
Restricting  $H$ to complete graphs, one can note that, 
given an input graph $G$, determining whether $G$ admits a $k$-coloring 
is NP-complete for all $k \geq 3$, and is polynomial time solvable otherwise 
(Theorem~\ref{th KM complexity}(ii) is an analogue of this observation). 
Thus, it makes sense to 
provide easily verifiable necessary conditions for a graph $G$ 
(resp., oriented graph $\overrightarrow{G}$)  being $k$-colorable 
(resp., pushably $k$-colorable) for $k \geq 3$. 

One of the best known general results for simple graphs that 
provides such a necessary condition 
for $k$-critical graphs with respect to its edge density is due to Kostochka and Yancy~\cite{kostochka2014ore} which (almost) solves Ore's conjecture. 
Restricted to $k=3$~\cite{kostochkagrotzch}, their result resolves the very first case of Ore's conjecture and provides an
alternative proof of Gr\"{o}tzsch's Theorem~\cite{kostochkagrotzch}. 
We state this result, 
which is one of the earliest examples of the usage of the ``potential method'' in the theory of coloring, in the following. 

\begin{theorem}[\cite{kostochkagrotzch}]\label{th KY k=4}
    Let $G$ be a $3$-critical graph. Then 
    $|E(G)| \geq \frac{5|V(G)|-2}{3}$. 
\end{theorem}

The analogous result of the above theorem in the context of signed (circular) chromatic number is presented in~\cite{beaudou2024density}. Our main result, Theorem~\ref{th push 3-critical density}, is the analogue for oriented graphs with respect to pushable homomorphism.

\medskip

Given a family $\mathcal{F}$ of simple graphs, its oriented 
and pushable chromatic numbers are defined as 
$$\chi_o(\mathcal{F}) = \max\{\chi_o(\overrightarrow{G}) : G \in \mathcal{F}\} \text{ and }  \chi_p(\mathcal{F}) = \max\{\chi_p(\overrightarrow{G}) : G \in \mathcal{F}\}.$$

The \textit{maximum average degree} 
of a graph $G$ (resp., an oriented graph $\overrightarrow{G}$) is given by 
$$mad(\overrightarrow{G})=mad(G) = \max \left\{\frac{2|E(H)|}{|V(H)|}: H \text{ is a subgraph of } G\right\}.$$

The initial studies on oriented chromatic number, following its introduction in 1994 due to Courcelle~\cite{courcelle-monadic}, revolved around
finding $\chi_o(\mathcal{P}_g)$ for $g \geq 3$, where $\mathcal{P}_g$ denotes the family 
of planar graphs having girth at least $g$. 
This line of study is motivated by finding analogous versions of 
the Four-Color 
Theorem~\cite{appel1989every} and Gr\"{o}tzsch's Theorem~\cite{kostochkagrotzch} for oriented graphs. 
Even though lower and upper bounds of $\chi_o(\mathcal{P}_g)$ exist for all values of $g \geq 3$, the exact values of $\chi_o(\mathcal{P}_g)$ 
are known only for $g \geq 12$. Specifically, 
$\chi_o(\mathcal{P}_{g}) = 5$~\cite{borodin2007oriented} for all $g \geq 12$. Finding the 
exact values of $\chi_o(\mathcal{P}_g)$ for all $g$, where $3 \leq g \leq 11$,  are
long standing open problems (the last known improvement was in 2012, see survey~\cite{sopena_survey} for details).

 Later in 2004, 
Klostermeyer and MacGillivray~\cite{push} introduced a modified version of the oriented 
coloring and chromatic number, namely, the pushable coloring and chromatic number, and, 
not surprisingly, finding the values of $\chi_p(\mathcal{P}_g)$ was given a special focus. 
To date, exact values of  $\chi_p(\mathcal{P}_g)$ are known only for $g = 8, 9$  and $g \geq 16$. Specifically, it is known that $\chi_o(\mathcal{P}_{8})=\chi_o(\mathcal{P}_{9}) = 4$~\cite{pascalpush, sen-push} and 
$\chi_o(\mathcal{P}_{g}) = 3$~\cite{borodin1998universal}, for $g \geq 16$. 
Finding the 
exact values of 
$\chi_p(\mathcal{P}_g)$ for $g \in \{3, 4, \cdots, 7, 10, 11, \cdots, 15\}$ are also
long standing open problems (the last known improvement was in 2017~\cite{sen-push}).

It is 
well-known~\cite{borodin1999maximum} that a planar graph (resp., projective planar graph) $G$ with girth $g$ 
satisfies $mad(G) < \frac{2g}{g-2}$. In practice, several upper bounds of 
$\chi_o(\mathcal{P}_g)$ and $\chi_p(\mathcal{P}_g)$ has been derived using this result,
that is, via establishing upper bounds for 
$\chi_o(\mathcal{M}_{\frac{2g}{g-2}})$ 
and $\chi_p(\mathcal{M}_{\frac{2g}{g-2}})$, where $\mathcal{M}_t$ denotes the family of graphs having maximum average degree strictly less than $t$~\cite{bensmail2021pushable,borodin1998universal,borodin1999maximum, das2023pushable,sen-push}.

In 1998, Borodin, Kostochka, Ne\v{s}et\v{r}il, Raspaud, and Sopena~\cite{borodin1998universal} 
made an interesting study of homomorphically embedding sparse (with respect to bounded maximum average degree, girth restrictions and planarity) oriented graphs into planar graphs. One of their main results proved that 
any oriented graph having maximum average degree strictly less than 
$\frac{16}{7}$ and girth at least $11$ is $AT(\overrightarrow{C}_3)$-colorable, equivalently, pushably $3$-colorable (by Observation~\ref{obs C3=3}). However, the tightness of this bound, both in terms of maximum average degree and girth, was left open. 

\begin{theorem}[\cite{borodin1998universal}]~\label{th borodin mad push}
    Let $\overrightarrow{G}$ be an oriented graph having maximum average degree strictly less than
    $\frac{16}{7}$ and girth at least $11$. Then we have
    \begin{enumerate}[(i)]
        \item $\chi_p(\overrightarrow{G}) \leq 3$,

        \item $\overrightarrow{G} \rightarrow AT(\overrightarrow{C}_3)$. 
    \end{enumerate}
\end{theorem}

Since all planar and projective planar graphs having girth at least $16$ have maximum average degree strictly less than $\frac{16}{7}$, an immediate corollary follows. 

\begin{corollary}[\cite{borodin1998universal}]~\label{cor borodin planar 16}
   Let $\overrightarrow{G}$ be an oriented planar or projective planar graph
    having girth at least $16$. Then we have 
    \begin{enumerate}[(i)]
        \item $\chi_p(\overrightarrow{G}) \leq 3$,

        \item $\overrightarrow{G} \rightarrow AT(\overrightarrow{C}_3)$. 
    \end{enumerate} 
\end{corollary}

We improve Theorem~\ref{th borodin mad push} by proving the following result, using Theorem~\ref{th push 3-critical density}. Both the maximum average degree and the girth condition in our result are tight. 

\begin{theorem}\label{th mad push}
    Let $\overrightarrow{G}$ be an oriented graph having maximum average degree strictly less than
    $\frac{30}{13}$ and girth at least $5$. Then we have
    \begin{enumerate}[(i)]
        \item $\chi_p(\overrightarrow{G}) \leq 3$,

        \item $\overrightarrow{G} \rightarrow AT(\overrightarrow{C}_3)$. 
    \end{enumerate}
    Moreover, it is not possible to relax either the maximum average degree or the girth conditions. 
\end{theorem}    

\begin{proof}
 (i) Suppose the contrary, and let $\overrightarrow{G}$ be a pushably $3$-critical oriented graph having maximum average degree strictly less than $\frac{30}{13}$ and girth at least $5$. 
 Notice that 
 $\overrightarrow{G} \neq \overrightarrow{E}_i$ since
 $mad(\overrightarrow{E}_i) = \frac{30}{13}$ for all $i \in \{1, 2, 3\}$. 
 Moreover, $\overrightarrow{G} \neq \overrightarrow{C}_{-4}$ as $\overrightarrow{G}$ has girth at least $5$. Therefore, according to 
 Theorem~\ref{th push 3-critical density}, $\overrightarrow{G}$ must have at least $\frac{15|V(\overrightarrow{G})|+2}{13}$ arcs, implying 
 $mad(\overrightarrow{G}) \geq \frac{30}{13}$, a contradiction.  \\

 \noindent (ii) The proof follows from Observation~\ref{obs C3=3} and Theorem~\ref{th mad push}(i).
\end{proof}

   Notice that the only pushably $3$-critical oriented graphs having maximum average degree equal to $\frac{30}{13}$ are 
        $\overrightarrow{E}_1,$ 
        $\overrightarrow{E}_2,$ and
        $\overrightarrow{E}_3,$ 
        (see Fig.~\ref{fig:Ei}). The following corollary of Theorem~\ref{th mad push} is an improvement of Corollary~\ref{cor borodin planar 16}.

\begin{corollary}~\label{cor planar 15}
    Let $\overrightarrow{G}$ be an oriented planar or projective planar graph
    having girth at least $15$. Then we have 
    \begin{enumerate}[(i)]
        \item $\chi_p(\overrightarrow{G}) \leq 3$,

        \item $\overrightarrow{G} \rightarrow AT(\overrightarrow{C}_3)$. 
    \end{enumerate} 
\end{corollary}

\begin{proof}
    This follows directly from Theorem~\ref{th mad push} and the known fact~\cite{borodin1999maximum} that planar and projective planar graphs having girth at least $15$ have maximum average degree strictly less than 
    $\frac{30}{13}$. 
\end{proof}

\begin{rem}
    Corollary~3.4 of~\cite{borodin2004homomorphisms}  claims that if $\overrightarrow{G}$ is a planar graph  having girth at least $13$, then $\overrightarrow{G} \rightarrow AT(\overrightarrow{C}_3)$. This claim, if correct, is clearly stronger than the above corollary (for planar graphs). However, unfortunately, the proof of Corollary~3.4 has an error and it has been acknowledged, and confirmed by the authors of~\cite{borodin2004homomorphisms} via an email conversation. Thus, our result (Corollary~\ref{cor planar 15}) is the best known in this line of work as of now. That said, the claim in Corollary~3.4~\cite{borodin2004homomorphisms} has not been disproven and may still hold, although proving it seems challenging. This remark can be considered an \textit{erratum} to Corollary~3.4 of~\cite{borodin2004homomorphisms}. 
\end{rem}

Corollary~\ref{cor planar 15} finds the exact value $\chi_p(\mathcal{P}_{15}) = 3$, 
and thus closes one of the open cases mentioned in this section. On the other hand, it is known~\cite{pascalpush,sen-push} that 
$\chi_p(\mathcal{P}_8) = \chi_p(\mathcal{P}_9) = 4$. Thus, answering the following problem will solve all the open cases of finding the exact value of $\chi_p(\mathcal{P}_{g})$ for $g \geq 8$. 

\begin{problem}
Find the minimum $g \in \{10, 11, 12, 13, 14\}$ such that  
$\chi_p(\mathcal{P}_{g}) = 3$ and $\chi_p(\mathcal{P}_{g-1}) = 4$.   
\end{problem}

\medskip

The $L(p,q)$-labeling~\cite{gallian2022dynamic} is a popular graph-theoretic model for the Channel Assignment Problem in wireless networks.  
A few articles have studied its oriented analogue, which may better approximate the real-life scenario since several transmissions are one-way~\cite{calamoneri20132,chang2007distance,gonccalves2006oriented,sen20142}. The two variants of the oriented analogue of $L(p,q)$-labeling are as follows. 

A \textit{$2$-dipath $\ell$-$L(p,q)$-labeling}~\cite{calamoneri20132} of an oriented graph $\overrightarrow{G}$ is a function 
$g : V(\overrightarrow{G}) \rightarrow \{0,1,\ldots,\ell\}$ 
satisfying
\begin{enumerate}[(i)]
    \item If $u,v$ are adjacent, then $|\ell(u) - \ell(v)| \geq p$.

    \item If $u,v$ are endpoints of a directed $2$-path (a directed path with two arcs), then 
    $|\ell(u) - \ell(v)| \geq q$.
\end{enumerate}
The \textit{$2$-dipath $L(p,q)$-labeling span} of $\overrightarrow{G}$, denoted by $\overrightarrow{\lambda}_{p,q}(\overrightarrow{G})$, 
is the minimum $\ell$ such that $\overrightarrow{G}$ admits a $2$-dipath $\ell$-$L(p,q)$-labeling.

An 
\textit{oriented $\ell$-$L(p,q)$-labeling}~\cite{gonccalves2006oriented} of an oriented graph $\overrightarrow{G}$ is a 
$2$-dipath $\ell$-$L(p,q)$-labeling
which is also an oriented coloring. 
The \textit{oriented $L(p,q)$-labeling span} of $\overrightarrow{G}$, denoted by 
$\lambda^o_{p,q}(\overrightarrow{G})$, 
is the minimum $\ell$ such that $\overrightarrow{G}$ admits an oriented $\ell$-$L(p,q)$-labeling.

Given a family $\mathcal{F}$ of simple graphs, its $2$-dipath  and oriented $L(p,q)$-labeling spans are defined as 
$$\overrightarrow{\lambda}_{p,q}(\mathcal{F}) = \max\{\overrightarrow{\lambda}_{p,q}(\overrightarrow{G}) : G \in \mathcal{F}\} \text{ and }  \lambda^o_{p,q}(\mathcal{F}) = \max\{\lambda^o_{p,q}(\overrightarrow{G}) : G \in \mathcal{F}\}.$$

\begin{theorem}\label{thm: lpq}
   Let $\overrightarrow{G}$ be an oriented graph. 
   If $\overrightarrow{G}$ has maximum average degree 
    $\frac{30}{13}$ and girth at least $5$, or 
    if $G$ is a planar or a projective planar graph
    having girth at least $15$, then we have
    $\overrightarrow{\lambda}_{p,q}(\overrightarrow{G}) \leq \lambda^o_{p,q}(\overrightarrow{G}) \leq 2p+3q$, where $p \geq q$. 
    \end{theorem}

\begin{proof}

  It is known~\cite{sen20142} that 
  $\overrightarrow{\lambda}_{p,q}(\overrightarrow{G}) \leq \lambda^o_{p,q}(\overrightarrow{G})$ and that 
    $\lambda^o_{p,q}(\overrightarrow{G}) \leq 
    \lambda^o_{p,q}(\overrightarrow{H})$
    if $\overrightarrow{G} \rightarrow \overrightarrow{H}$. Let us assume that the set of vertices of the 
directed $3$-cycle $\overrightarrow{C}_3$ is the  additive cyclic group 
$\mathbb{Z}/3\mathbb{Z} = \{\overline{0},\overline{1},\overline{2}\}$ 
while the arcs are of the form $\overline{i}~\overline{(i+1)}$ where 
$\overline{i} \in \mathbb{Z}/3\mathbb{Z}$. Now, consider the function 
    $$f: V(AT(\overrightarrow{C}_3)) \to \{0, 1, 2, \ldots, (2p+3q)\}$$
    given by 
    $f(\bar{0}) = 0$, 
    $f(\bar{0}') = q$,
    $f(\bar{1}) = p+q$,
    $f(\bar{1}') = p+2q$,
    $f(\bar{2}) = 2p+2q$,
    $f(\bar{2}') = 2p+3q$.
    Observe that $f$ is an oriented $(2p+3q)$-$L(p,q)$-labeling
    of $AT(\overrightarrow{C}_3)$ for all $p \geq q$.
    Thus, the proof follows from Theorem~\ref{th mad push}(ii)
    and Corollary~\ref{cor planar 15}(ii).
\end{proof}

Sen~\cite{sen20142} showed that 
$\overrightarrow{\lambda}_{2,1}(\mathcal{P}_{16}) \leq \lambda^o_{2,1}(\mathcal{P}_{16}) \leq 7$. Note that Theorem~\ref{thm: lpq}, for $(p,q)=(2,1)$, improves these results.

\medskip

In the next section, we will present the proof of 
Theorem~\ref{th push 3-critical density}. The proof uses potential method, and to the best of our knowledge, this is the first instance where potential method is applied to prove results in the context of oriented 
(resp., pushable) coloring. The application of the potential method 
for oriented graphs is also another important feature of this article.

\section{Proof of Theorem~\ref{th push 3-critical density}}\label{sec proof of the main theorem}

\begin{figure}
    \centering

                \begin{tikzpicture}[scale=0.6]
    \node[draw, circle, minimum size=.75cm] (v1) at (0,0) {$v_1$};
    \node[draw, circle, minimum size=.75cm] (v2) at (8,0) {$v_2$};
    \node[draw, circle, minimum size=.75cm] (v3) at (4,6) {$v_3$};
    \node[draw, circle, minimum size=.75cm] (v4) at (4,2.5) {$v_4$};

    \node[draw, circle, minimum size=.002cm] (v14) at (4,4.2) {};
    \node[draw, circle, minimum size=.002cm] (v7) at (4,0) {};
    \node[draw, circle, minimum size=.002cm] (v12) at (2,0) {};
    \node[draw, circle, minimum size=.002cm] (v13) at (6,0) {};

    \node[draw, circle, minimum size=.002cm] (v8) at (1.5,0.9) {};
    \node[draw, circle, minimum size=.002cm] (v9) at (2.7,1.68) {};
    \node[draw, circle, minimum size=.002cm] (v10) at (5.3,1.68) {};
    \node[draw, circle, minimum size=.002cm] (v11) at (6.5,0.9) {};
    

    \draw[->] (v3) -- (v1);
    \draw[->] (v3) -- (v2);
    \draw[->] (v4) -- (v14);
    \draw[->] (v14) -- (v3);
    
    \draw[->] (v13) -- (v2);
    \draw[->] (v7) -- (v13);
    \draw[<-] (v12) -- (v7);
    \draw[<-] (v12) -- (v1);
    
    \draw[->] (v1) -- (v8);
    \draw[->] (v8) -- (v9);
    \draw[->] (v9) -- (v4);
    
    \draw[->] (v2) -- (v11);
    \draw[<-] (v10) -- (v11);
    \draw[->] (v10) -- (v4);


\end{tikzpicture}

    \caption{A pushably $3$-critical oriented graph $\overrightarrow{F}$ on $12$ vertices and $14$ arcs. The potential of this oriented graph is $-2$.}
    \label{fig:tightness}
\end{figure}
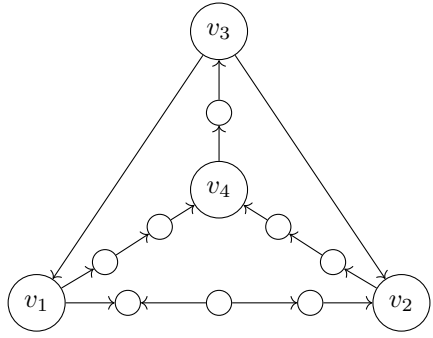

We prove Theorem~\ref{th push 3-critical density} by the method of contradiction using  the discharging and the potential methods. Before we begin, we note that the tightness of Theorem~\ref{th push 3-critical density} follows from the example of the pushably $3$-critical graph presented in Fig.~\ref{fig:tightness}.  Given an oriented graph $\overrightarrow{G}$, we define a potential function 
$$\rho(\overrightarrow{G}) = 15|V(\overrightarrow{G})| - 13|A(\overrightarrow{G})|.$$ 
Observe that to prove Theorem~\ref{th push 3-critical density} we can equivalently show that if $\overrightarrow{G}$ is a pushably $3$-critical oriented graph other than $\overrightarrow{C}_{-4}, \overrightarrow{E}_1, \overrightarrow{E}_2$, and $\overrightarrow{E}_3$, then $\rho(\overrightarrow{G}) \leq -2$. Notice that, $\overrightarrow{C}_{-4}$, $\overrightarrow{E}_1, \overrightarrow{E}_2,$ and $\overrightarrow{E}_3$, are all pushably $3$-critical oriented graphs and none of them have potential less than or equal to $-2$.
Since our proof is by contradiction, we will assume that 
there exist pushably $3$-critical oriented graphs, other than $\overrightarrow{C}_{-4}$, $\overrightarrow{E}_1, \overrightarrow{E}_2,$ and $\overrightarrow{E}_3$, having potential strictly greater than $-2$, that is, greater than or equal to $-1$, such that among all such examples,  $\overrightarrow{M}$ is minimal with respect to 
$|V(\overrightarrow{M})|+|A(\overrightarrow{M})|$. 
  For the rest of the proof, we build structural properties of $\overrightarrow{M}$ and show that it can not exist.

That means $\overrightarrow{M} \neq \overrightarrow{C}_{-4}$,
$\overrightarrow{M} \neq \overrightarrow{E}_i$ for all $i \in \{1, 2, 3\}$, $\overrightarrow{M}$ does not admit a pushable homomorphism to 
$\overrightarrow{C}_3$ while all its proper subgraphs do, and 
$\rho(\overrightarrow{M}) \geq -1$.
Moreover, if $\rho(\overrightarrow{G})\geq -1$ and $|V(\overrightarrow{G})|+|A(\overrightarrow{G})| < |V(\overrightarrow{M})|+|A(\overrightarrow{M})|$, then either $\overrightarrow{G}$ is $\overrightarrow{C}_{-4}$ or $\overrightarrow{E}_i$ for some $i \in \{1,2,3\}$, or it is not pushably $3$-critical.

\begin{observation}\label{potential value}
    Let $\overrightarrow{K}_n$ denote an orientation of the complete graph $K_n$ on $n$ vertices
    and $\overrightarrow{K_n - e}$ denote an orientation of the graph $K_n - e$ obtained by deleting 
    one edge from $K_n$. Then we have: 
    $\rho(\overrightarrow{K}_1)=15$, 
    $\rho(\overrightarrow{K}_2)=17$,
    $\rho(\overrightarrow{K}_3)=6$,
    $\rho(\overrightarrow{K_3 - e})=19$,
    $\rho(\overrightarrow{C}_{-4})=8$, and
    $\rho(\overrightarrow{E_i})=0$, for all $i \in \{1, 2, 3\}$. 
\end{observation}

Let $\overrightarrow{G}$ be an oriented graph.  Let $\overrightarrow{P}_n$ be an oriented path with $n$ arcs having endpoints $x$ and $y$. The oriented graph $P_n(\overrightarrow{G})$ is 
obtained by adding $\overrightarrow{P}_n$ to $\overrightarrow{G}$ and identifying $x$ and $y$ with some (not necessarily distinct) vertices  of 
$\overrightarrow{G}$. 

\begin{lemma}[Gap Lemma]\label{lem gap-lemma}
    Let $\overrightarrow{H}$ be a subgraph of $\overrightarrow{M}$. Then we have
    \begin{enumerate}[(i)]
        \item $\rho(\overrightarrow{H}) \geq -1$, if $\overrightarrow{H}=\overrightarrow{M}$,

        \item $\rho(\overrightarrow{H}) \geq 6$, if $\overrightarrow{H} \equiv_p \overrightarrow{C}_3$ or if
        $\overrightarrow{M} = P_4(\overrightarrow{H})$,

        \item $\rho(\overrightarrow{H}) \geq 7$, otherwise. 
    \end{enumerate}
\end{lemma}

\begin{proof}
    Let $\overrightarrow{H}$ be a maximal 
    (with respect to $|V(H)| + |A(H)|$) counterexample to Lemma~\ref{lem gap-lemma}. Firstly, if $\overrightarrow{H} = \overrightarrow{M}$, 
    then $\rho(\overrightarrow{H}) \geq -1$ by our assumption. 
    Moreover, if 
    $\overrightarrow{H} \equiv_p \overrightarrow{C}_3 = \overrightarrow{K}_3$, then 
    $\rho(\overrightarrow{H}) =6$  by Observation~\ref{potential value}. 
    Furthermore, since $P_k(\overrightarrow{H})$ 
    has $k-1$ vertices and $k$ arcs more than $\overrightarrow{H}$, we have
    \begin{equation}\label{eq Pk-H}
      \rho(P_k(\overrightarrow{H})) = \rho(\overrightarrow{H}) + (15 \times (k-1)) - (13 \times k) = \rho(\overrightarrow{H})  +2(k-1) - 13.  
    \end{equation}
Thus, in particular, if $P_k(\overrightarrow{H}) = \overrightarrow{M}$ for $k \in \{2,3,4\}$, then using Equation~(\ref{eq Pk-H}) we have $\rho(\overrightarrow{H}) \geq 6$, with the inequality being strict for $k \in \{2,3\}$.

 So, if $\overrightarrow{H}$ requires to be a maximal counter example, then $\overrightarrow{H} \neq \overrightarrow{M}$, 
 $\overrightarrow{H} \not\equiv_p \overrightarrow{C}_3$, and  
    $P_4(\overrightarrow{H}) \neq \overrightarrow{M}$. 
    Moreover, $\overrightarrow{H}$ is a proper subgraph of $\overrightarrow{M}$ satisfying $\rho(\overrightarrow{H}) \leq 6$.

Using Observation~\ref{obs C3=3} we can say that any oriented 
graph  having $3$ or less vertices, other than 
$\overrightarrow{C}_3$ (up to  push equivalence)
have potential at least $7$. Hence, 
we may assume $\overrightarrow{H}$ to have at least 
$4$ vertices. Since $\overrightarrow{H}$ is a proper 
subgraph of $\overrightarrow{M}$, it must admit a 
pushable homomorphism to $\overrightarrow{C}_3$ by 
Observation~\ref{obs C3=3} as $\overrightarrow{M}$ is 
pushably $3$-critical. 

 This implies the existence of a homomorphism $f: \overrightarrow{H}^S \rightarrow \overrightarrow{C}_3$  for some 
 $S \subseteq V(\overrightarrow{H})$. 
 Without loss of generality we may replace $\overrightarrow{M}^S$ (resp., $\overrightarrow{H}^S$) 
 with $\overrightarrow{M}$ (resp., $\overrightarrow{H}$) to simplify the notation. Thus, we may assume that 
 $f$ is a homomorphism of $\overrightarrow{H}$ to 
 $\overrightarrow{C}_3$.

 Next, replace the vertices of 
 $\overrightarrow{H}$ with the vertices of $\overrightarrow{C}_3$ in $\overrightarrow{M}$ to obtain a graph $\overrightarrow{M}'$. 
 If a vertex $u \in V(\overrightarrow{M}) \setminus V(\overrightarrow{H})$ is adjacent to a vertex $v \in V(\overrightarrow{H})$ in $\overrightarrow{M}$, then in $\overrightarrow{M}'$ we make $u$ and $f(v)$ adjacent. Moreover, we match the direction of the arc between $u$ and $f(v)$ with the direction of the arc between $u$ and $v$. For convenience, the set of vertices of the above mentioned $\overrightarrow{C}_3$ is denoted by 
 $\{\overline{0},\overline{1},\overline{2}\}$. 

By Equation~(\ref{eq Pk-H}) and from our assumption, we know that 
$\rho(P_k(\overrightarrow{H})) \leq -3$ for $k \in \{2, 3\}$. Thus,  $\overrightarrow{M}$ cannot contain a $P_k(\overrightarrow{H})$ as a subgraph for $k \in \{2,3\}$ as it contradicts the maximality of $\overrightarrow{H}$.
This implies that $\overrightarrow{M}'$ must be an oriented graph  and does not contain $\overrightarrow{C}_{-4}$ as a subgraph.

Let $$f_{ext}(x) = 
\begin{cases}
f(x) & \text{ if } x \in V(\overrightarrow{H}),\\
   x & \text{ if } x \in V(\overrightarrow{M}) \setminus V(\overrightarrow{H}).
\end{cases}
$$
Note that $f_{ext}$ is a homomorphism of $\overrightarrow{M}$ to $\overrightarrow{M}'$. Therefore, 
if there exists a pushable homomorphism 
$g: \overrightarrow{M}' \xrightarrow{push} \overrightarrow{C}_3$, then 
$g \circ f_{ext}$ is a pushable homomorphism of 
$\overrightarrow{M}$ to $\overrightarrow{C}_3$ by composition. However, since $\overrightarrow{M}$ is pushably $3$-critical, it can not admit a pushable homomorphism to $\overrightarrow{C}_3$. Hence, there does not exist any pushable homomorphism of $\overrightarrow{M}'$ to $\overrightarrow{C}_3$. Thus, $\overrightarrow{M}'$ must contain a pushably $\overrightarrow{C}_3$-critical 
(equivalently, pushably $3$-critical by Observation~\ref{obs C3=3}) oriented subgraph $\overrightarrow{M}''$.

Let $X = \{\overline{0},\overline{1},\overline{2}\} \cap V(\overrightarrow{M}'')$.  On the one hand,  if $X = \emptyset$, then $\overrightarrow{M}''$ is a proper subgraph of $\overrightarrow{M}$, which is impossible since both are pushably $3$-critical. On the other hand, 
if $V(\overrightarrow{M}'') \subseteq X$, then $\overrightarrow{M}''$ is a subgraph of $\overrightarrow{C}_3$, and thus cannot be $\overrightarrow{C}_3$-critical. Therefore, $\overrightarrow{M}''$ has some vertices inside $X$ and some outside, that is, $V(\overrightarrow{M}'') \setminus X \neq \emptyset$ and $V(\overrightarrow{M}'') \cap X \neq \emptyset.$

Now, we consider the set 
$$Y = (V(\overrightarrow{M}'') \setminus X) \cup V(\overrightarrow{H)}.$$ Let $\overrightarrow{H}'$ be the oriented subgraph of $\overrightarrow{M}$ induced by $Y$. Since $V(\overrightarrow{M}'') \setminus X$ is non-empty, $\overrightarrow{H}'$ must 
have more vertices and arcs than $\overrightarrow{H}$. 
In particular, the number of vertices and arcs in $\overrightarrow{H}'$ are
$$|V(\overrightarrow{H}')| = |V(\overrightarrow{M}'')| - |X| + |V(\overrightarrow{H})|$$
and  
$$|A(\overrightarrow{H}')| = |A(\overrightarrow{M}'')| - |A(\overrightarrow{M}'[X])| - A[V(\overrightarrow{M}'') \setminus X, X] +  
A[V(\overrightarrow{M}'') \setminus X, V(\overrightarrow{H})] + |A(\overrightarrow{H})|.$$
Since $\overrightarrow{M}$ does not contain any $P_2(\overrightarrow{H})$ as a subgraph, we have 
$$A[V(\overrightarrow{M}'') \setminus X, X] \leq  
A[V(\overrightarrow{M}'') \setminus X, V(\overrightarrow{H})],$$
and thus we have 
$$|A(\overrightarrow{H}')| \geq |A(\overrightarrow{M}'')| - |A(\overrightarrow{M}'[X])| + |A(\overrightarrow{H})|.$$

Hence, 
\begin{equation}\label{eq gap calculation}
\rho(\overrightarrow{H}') \leq \rho(\overrightarrow{M}'')
- \rho(\overrightarrow{M}'[X]) + \rho(\overrightarrow{H})
\end{equation}
Note that, $\rho(\overrightarrow{M}'') \leq 0$, and unless $\overrightarrow{M}'' \equiv_p \overrightarrow{E}_i$ for some $i \in \{1,2,3\}$, $\rho(\overrightarrow{M}'') \leq -2$ as otherwise it will contradict the minimality of $\overrightarrow{M}$. 
 On the other hand, $\rho(\overrightarrow{M}'[X]) = 6$ if $\overrightarrow{M}'[X] \equiv_p \overrightarrow{C}_3$ and $\rho(\overrightarrow{M}'[X]) \geq 15$ otherwise. 
Moreover, $\rho(\overrightarrow{H}) \leq 6$ according to our assumption.  Observe that all $E_i$'s are triangle-free. Thus, when $\rho(\overrightarrow{M}'[X]) = 6$, we have 
$\rho(\overrightarrow{M}'') \leq -2$. 
Therefore, \\
\[ \rho(\overrightarrow{H}') \leq
\begin{cases}
     -2 - 6 + 6 = -2,  ~\text{if}~ \overrightarrow{M}'[X] \equiv_p \overrightarrow{C}_3, \\
     0 - 15 + 6 = -9, ~\text{otherwise.}~
\end{cases}\]

Thus, no matter what, $\rho(\overrightarrow{H}') \leq -2$, and hence 
$\overrightarrow{H}'$ contradicts the maximality of $\overrightarrow{H}$ and this completes the proof of the lemma.
\end{proof}

\subsection{List of reducible configurations of $\overrightarrow{M}$}

Given a graph $G$, a vertex of degree exactly 
(resp., at least, at most) $k$ is called a 
\emph{$k$-vertex (resp., $k^+$-vertex, $k^-$-vertex)}. 
A \textit{$k$-chain} of $G$ is a path having $k+1$ vertices with  $3^+$-vertices as endpoints and 
$2$-vertices as internal vertices. 
The endpoints of a $k$-chain are \emph{chain-adjacent} to each other and 
the internal vertices and arcs of a $k$-chain are \textit{chain-incident} to the endpoints. 
If a $k$-vertex (for some $k \geq 3$) 
is chain-adjacent to exactly (resp., at least, at most) $t$  $2$-vertices, 
then we call it a \textit{$k^{t}$-vertex  (resp., $k^{\geq t}$-vertex, $k^{\leq t}$-vertex}). 
Furthermore, suppose $t=(t_1+t_2+\cdots+t_k)$ and a 
$k^t$-vertex $v$ has exactly $t_i$ many $2$-vertices in its $i^{th}$ chain (with respect to some pre-defined indexing of the $k$ chains incident with $v$). In this case, we say that $v$ is a $k_{t_1, t_2, \ldots, t_k}$-vertex. 
If we use such notation to  describe any vertex $v$ of an oriented graph 
$\overrightarrow{G}$,
we are actually meaning the description of $v$ in the underlying graph $G$ of $\overrightarrow{G}$. 
In some of the figures, where we present the list of ``reducible configurations'' (formally defined later), we use such notation inside a ``shaded circle'' to denote a vertex having such a property.

Let $\overrightarrow{G}$ be an oriented graph with a subset $X \subseteq V(\overrightarrow{G})$. Let $f: V(\overrightarrow{G}) \setminus X \to V(\overrightarrow{C}_3)$ be a function. 
We call the vertices of $V(\overrightarrow{G}) \setminus X$ as \textit{colored vertices}, and the vertices of $X$ as \textit{uncolored vertices}.  
Suppose there exists a push equivalent orientation $\overrightarrow{G}^S$ of $\overrightarrow{G}$
such that for any  arc $uv$ of 
$\overrightarrow{G}^S[V(\overrightarrow{G}) \setminus X]$, there is an arc from $f(u)$ to $f(v)$ in 
$\overrightarrow{C}_3$. In such a scenario, we say that $f$ is a 
\textit{partial pushable homomorphism} of 
$\overrightarrow{G}$ to $\overrightarrow{C}_3$. 
Equivalently, we say that $f$ is a 
\textit{partial homomorphism} of $\overrightarrow{G}^S$ to $\overrightarrow{C}_3$.
In practice, whenever we start with a partial pushable homomorphism 
$f$ of $\overrightarrow{G}$ to $\overrightarrow{C}_3$, 
we assume without loss of generality that $\overrightarrow{G}$ has an orientation  
such that $f$ is a partial homomorphism unless otherwise stated.

Given a partial (pushable) homomorphism 
$f: V(\overrightarrow{G}) \setminus X \to \overrightarrow{C}_3$, 
if it is possible to push a vertex subset $X' \subseteq X$ so that 
$f$ can be extended to a (pushable) homomorphism of 
$\overrightarrow{G}^{X'}$ to $\overrightarrow{C}_3$, then we say that $f$ is \textit{extendable}.

Let $\overrightarrow{P}$ be an \textit{oriented $(x,y)$-path}, that is, an oriented 
path starting at vertex $x$ and ending at vertex $y$.
The oriented path $\overrightarrow{P}$ is 
odd (resp., even) if it has odd (resp., even) number of forward arcs, and the particular orientation of $P$ is called an odd (resp., even) orientation. Here forward (resp., backward) arcs are determined with respect to traversal from $x$ to $y$. 
For an oriented $(x,y)$-path 
$\overrightarrow{P}$, 
let $f$ be a partial (pushable) homomorphism of $\overrightarrow{P}$ to 
$\overrightarrow{C}_3$ with $x$ and $y$ being the only colored vertices.  If $f$ is extendable (resp., not extendable), 
then we say that the color $f(x)$ at $x$ \textit{allows} (resp., \textit{forbids}) 
the color $f(y)$ at $y$. Since $\overrightarrow{C}_3$ is vertex transitive, it makes sense to not mention the color $f(x)$ at $x$ while noting the number of allowed or forbidden colors at $y$. 

\begin{table}[ht]
   \centering
    \begin{tabular}{|c|c|c|c|}
    \hline
        $k=$ & Path type & Allowed colors at $y$ & Forbidden colors at $y$\\
        \hline
        $1$ & even & $\overline{i+2}$ & $\overline{i}, \overline{i+1}$\\
        $2$ & even & $\overline{i+1}, \overline{i+2}$ & $\overline{i}$\\
        $3$ & even & $\overline{i}, \overline{i+1}$ & $\overline{i+2}$\\
        $4$ & even & $\overline{i}, \overline{i+1}, \overline{i+2}$ & $\phi$\\
        $5$ & even & $\overline{i}, \overline{i+1}, \overline{i+2}$ & $\phi$\\
        \hline
        $1$ & odd & $\overline{i+1}$ & $\overline{i}, \overline{i+2}$\\
        $2$ & odd & $\overline{i}$ & $ \overline{i+1}, \overline{i+2}$\\
        $3$ & odd & $\overline{i},  \overline{i+2}$ & $\overline{i+1}$\\
        $4$ & odd & $ \overline{i+1}, \overline{i+2}$ & $\overline{i}$\\
        $5$ & odd & $\overline{i}, \overline{i+1}, \overline{i+2}$ & $\phi$\\
        \hline
    \end{tabular}
    \caption{The presentation of Observation~\ref{obs path analysis}.}
    \label{table path analysis}
    \end{table}

\begin{observation}
Given any two odd (resp., even) orientations 
$\overrightarrow{P}$ and $\overrightarrow{P}'$, it is possible to obtain $\overrightarrow{P'}$ from $\overrightarrow{P}$ by pushing some internal vertices. 
\end{observation}
 Recall that the set of vertices of the 
directed $3$-cycle $\overrightarrow{C}_3$ is the elements of additive cyclic group 
$\mathbb{Z}/3\mathbb{Z} = \{\overline{0},\overline{1},\overline{2}\}$ 
while the arcs are of the form $\overline{i}~\overline{(i+1)}$ where 
$\overline{i} \in \mathbb{Z}/3\mathbb{Z}$.

\begin{observation}\label{obs path analysis}
    Let $\overrightarrow{P}_k$ be an oriented (x,y)-path on $k+1$ vertices. Then the color $\overline{i} \in \mathbb{Z}/3\mathbb{Z}$ at $x$ allows and forbids the set of colors at $y$ according to Table no~\ref{table path analysis}. 
 \end{observation}

Let $X \subseteq V(\overrightarrow{G})$ be a vertex subset of $\overrightarrow{G}$. 
If a vertex $v \in V(\overrightarrow{G}) \setminus X$ is adjacent to some vertex of $X$, then it is called a 
\textit{boundary point} of $X$. 
The \textit{closure} of $X$, denoted by $X^*$ and often referred to as a \textit{configuration}, is the oriented graph induced by $X$ and its boundary points. 
The set of all boundary points of $X$ (or $X^*$) is denoted by 
$\partial X^*$, and is called 
the \textit{boundary} of $X$ (or $X^*$). 
Finally, the $X^*$ is a \textit{reducible configuration} if any partial 
(pushable) homomorphism $f$ of $\overrightarrow{G}$, with the set of uncolored vertices $X$, is extendable.

\subsubsection{Using criticality}

\begin{lemma}\label{lem reducible-critical}
    The configurations listed below are reducible. See Fig.~\ref{fig:Ci} for a pictorial reference. 
\begin{enumerate}[(C1)]
    \item A $k$-chain for $k \geq 5$. 

    \item A $3^{\geq 7}$ vertex.
    
    \item A $4^{\geq 11}$ vertex. 


    \item A $3^5$-vertex adjacent to a 
    $3^5$-vertex. 

    \item A $3^6$-vertex adjacent to a 
    $4^{9}$-vertex. 

\item A $3^6$-vertex $1$-chain-adjacent to a 
$4^{\geq 9}$-vertex.

\item A $3^{\geq 5}$-vertex $1$-chain-adjacent to a $4^{10}$-vertex.




\item  A $3^4$-vertex adjacent to a $3^5$-vertex and $1$-chain-adjacent to a $4^{10}$-vertex.

\item  A $4^7$-vertex adjacent to a $3^6$-vertex and $1$-chain-adjacent to another $3^{6}$-vertex.




\item  A $3^{4}$-vertex $1$-chain-adjacent to a $3_{3,1,1}$-vertex and $1$-chain-adjacent to a $4^{10}$-vertex.

\item  A $4^{\geq7}$-vertex $1$-chain-adjacent to two $3^6$-vertices.




\item  A $4^{6}$-vertex $1$-chain-adjacent to three $3^6$-vertices.

\item A $4^{4}$-vertex $1$-chain-adjacent to four
$3^6$-vertices.

\item A directed $6$-cycle $\overrightarrow{C}$ of the form $xyu_1zu_2u_3x$ 
where $u_1, u_2, u_3$ are $2$-vertices, $x$ is a $3^5$-vertex, $y$ is a $3^{4}$-vertex, and $z$ is a $3^{3}$-vertex. 

\item A directed $6$-cycle $\overrightarrow{C}$ of the form $xyu_1zu_2u_3x$ 
where $u_1, u_2, u_3$ are $2$-vertices, $x$ is a $3^5$-vertex, $y$ is a $3^{\geq 2}$-vertex, and $z$ is a $3^{\geq 4}$-vertex.

\item A directed $6$-cycle $\overrightarrow{C}$ of the form $xu_1yu_2zu_3x$ 
where $u_1, u_2, u_3$ are $2$-vertices, $x$ is a $3^5$-vertex, 
$y$ is a $3^{\geq 4}$-vertex, and $z$ is a $3^{\geq 2}$-vertex. 

\end{enumerate} 
\end{lemma}

\begin{proof}
  For any of the above listed configurations, say \textbf{C}, there are some vertices for which all their neighbors are known and are part of the configuration. 
  For the other vertices, not all of the neighbors may
  be part of the configuration. Note that the latter vertices are the boundary vertices of the configuration 
  \textbf{C} and their set is denoted by $\partial \mathbf{C}$ .

  Our goal, for each configuration \textbf{C} listed above, is
  to assume an arbitrary partial pushable homomorphism 
  $f: V(\partial \mathbf{C}) \to V(\overrightarrow{C}_3)$ 
  and show that $f$ is extendable. Usually, the vertices of $\partial \mathbf{C}$ are independent, and thus, any function $f$ from $V(\partial \mathbf{C})$ to $\{\overline{0},\overline{1},\overline{2}\}$ is a partial pushable homomorphism. We will show such an arbitrary $f$ is extendable, primarily using 
  Observation~\ref{obs path analysis}, repeatedly. We list our proof of reducibility below enumerating with the name of the configurations.

 \begin{figure}[h!]
    \centering
    \includegraphics[scale=.7]{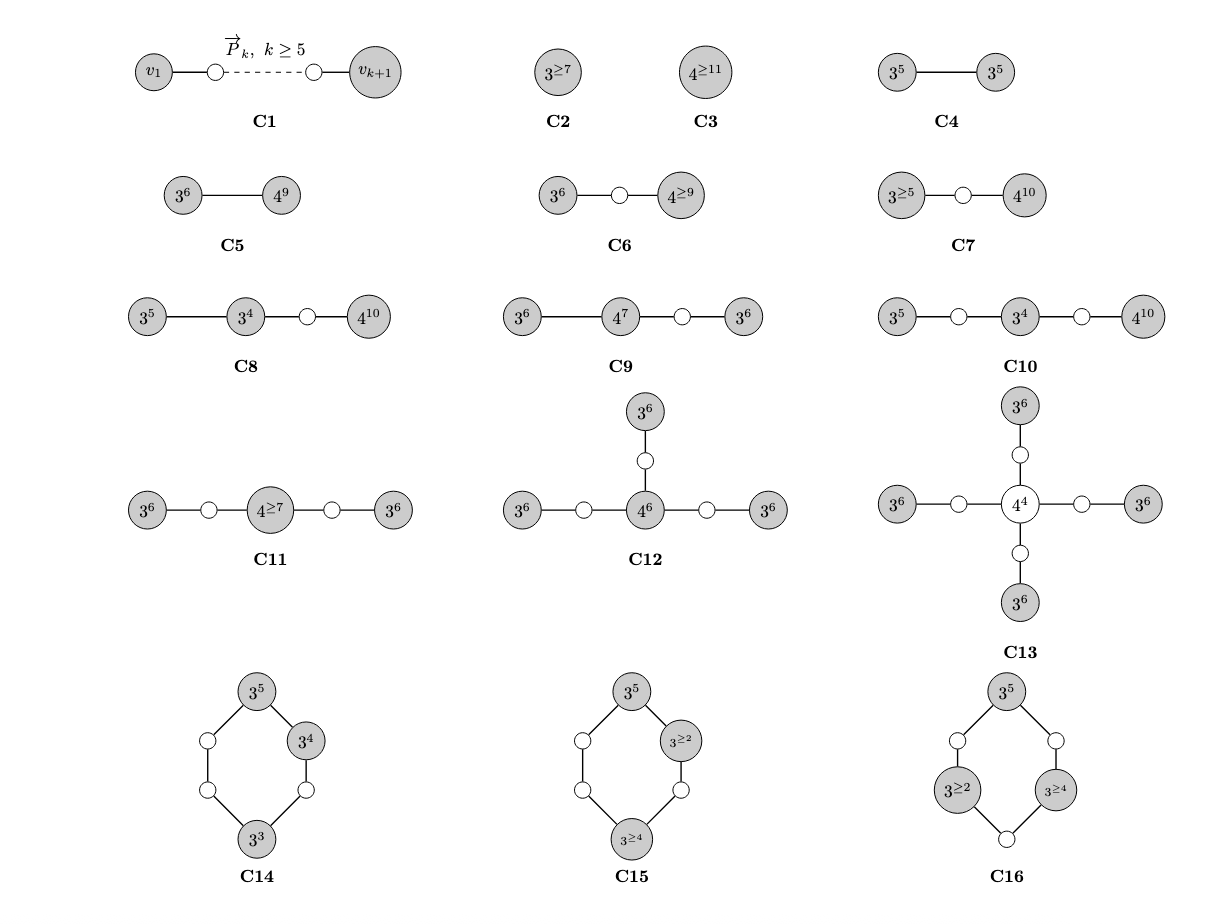}
       \caption{The list of reducible configurations (using criticality) for Lemma~\ref{lem reducible-critical}.}\label{fig:Ci} 
\end{figure}
  \noindent \textbf{(C1)}  According to 
  Observation~\ref{obs path analysis}, given any oriented $(x,y)$-path of length $5$ or more, $x$ forbids zero colors at $y$. Thus, \textbf{C1} is reducible. 
As \textbf{C1} is reducible, the longest chains we can have are 
of length $4$.~$\blacklozenge$ This inference will be used in reducing the subsequent configurations.

\medskip

\noindent \textbf{(C2)}  It is enough to show that a $3^7$-vertex $v$ is reducible. Let $v_1$ be a $3$-chain-adjacent vertex of $v$. Let $v_2,$ and $ v_3$ be the other chain-adjacent vertices of $v$. 
Now push $v$, if required, 
to make the $3$-chain connecting $v$ and $v_1$ even. According to Observation~\ref{obs path analysis}, $v_1$ allows all colors at $v$, while $v_2$ and $v_3$ forbid at most two colors at $v$. Thus, \textbf{C2} is reducible.~$\blacklozenge$

  \medskip

\noindent \textbf{(C3)}  It is enough to show that a $4^{11}$-vertex $v$ is reducible. Note that $v$ must be a $4_{3,3,3,2}$-vertex. Now push $v$, if required, 
to make two of the three incident $3$-chains even. 
According to Observation~\ref{obs path analysis}, the vertices of $\partial \mathbf{C3}$ forbid at most two colors at $v$. Thus, \textbf{C3} is reducible.~$\blacklozenge$

  \medskip

\noindent \textbf{(C4)} If required, push the two adjacent $3^5$-vertices, $u$ and $v$ (say), in such a way that one of their incident $3$-chains become even. Then the vertices of $\partial \mathbf{C4}$ allow
two colors at $u$ and two colors at $v$. Therefore, it is possible to extend $f$. Thus, \textbf{C4} is reducible.~$\blacklozenge$

  \medskip

\noindent \textbf{(C5)} Let $u$ and $v$ be the adjacent $3^6$-vertex and
$4^9$-vertex, respectively. If required, push $u,v$ in such a way that each of them have at most one incident odd $3$-chain. 
Then, the vertices of $\partial \mathbf{C5}$ allow 
two colors at $u$ and two colors at $v$. Therefore, it is possible to extend $f$. Thus, \textbf{C5} is reducible.~$\blacklozenge$

  \medskip

  \noindent \textbf{(C6)} It is enough to show that a $3^{6}$-vertex $u$ $1$-chain-adjacent to a $4^9$-vertex $v$ is reducible. Note that, $u$ must be a $3_{3,2,1}$-vertex and $v$ must be a $4_{3,3,2,1}$-vertex. Let $u$ be $3$-chain-adjacent to $u_1$, $2$-chain-adjacent to $u_2$ and $v$ be $3$-chain-adjacent to $v_1$ and $v_2$, and $2$-chain-adjacent to $v_3$. It is possible to make the $3$-chain and the $1$-chain incident to $u$ even, by pushing (if required) $u$ and $v$. We know that the $3$-chain (resp., $2$-chain) incident to $u$ forbids no (resp., at most one) color at $u$ by Observation~\ref{obs path analysis}.
  Hence, $u_1$ and $u_2$ allow at least two colors at $u$ and $u$ allows three colors at $v$. Therefore, it is possible to extend $f$ unless $v_1, v_2$ and $v_3$ forbid all three colors at $v$. If $v_1, v_2$ and $v_3$ forbid all colors at $v$ it means that the $3$-chains incident to $v$ are odd. In this case, we push the vertex $v$ making the $1$-chain incident to $v$ odd and the $3$-chains incident to $v$ even. Now $v_1, v_2$ and $v_3$ allow at least two colors at $v$ whereas $u$ forbids at most one color at $v$. Therefore, it is possible to extend $f$. Thus, \textbf{C6} is reducible. ~$\blacklozenge$

  \medskip

\noindent \textbf{(C7)} It is enough to show that a $3^{5}$-vertex $u$ $1$-chain-adjacent to a $4^{10}$-vertex $v$ is reducible. Note that, $u$ is either a $3_{3,1,1}$-vertex or a $3_{2,2,1}$-vertex, and $v$ must be a $4_{3,3,3,1}$-vertex. 
Let $v$ be $3$-chain-adjacent to $v_1, v_2$ and $v_3$ and $u$ be chain-adjacent to $u_1$ and $u_2$ apart from $v$. 
 It is possible to make two $3$-chains incident to $v$ even, by pushing $v$ (if required). Hence, $v_1, v_2$ and $v_3$ allow at least two colors at $v$. We push the vertex $u$ (if required) to make the $1$-chain connecting $u$ and $v$ even. 
Therefore, it is possible to extend $f$ unless $u_1$ and $u_2$ forbid all three colors at $u$. That can only happen when $u$ is a $3_{3,1,1}$-vertex and  the $3$-chain and the $1$-chain connecting $u$ to $u_1$ and $u_2$, respectively, are both odd.  In this case, we push the vertex $u$ making the $1$-chain connecting it to $v$ odd, and the other two incident chains even. 
Now $u_1$ and $u_2$ allow at least two colors at $u$ whereas $u$ forbids at most one color at $v$. That means, $u, v_1, v_2, v_3$ together forbid at most two colors at $v$. Therefore, it is possible to extend $f$.  Thus, \textbf{C7} is reducible.~$\blacklozenge$

  \medskip

\noindent \textbf{(C8)} Let $u$ be a $3^4$-vertex, $1$-chain-adjacent to a $4^{10}$-vertex $v$, and adjacent to a $3^5$-vertex $w$. 
Note that, $u$ must be a $3_{3,1,0}$-vertex, $v$ must be a $4_{3,3,3,1}$-vertex and $w$ be a $3_{3,2,0}$-vertex. 
Let $u$ be $3$-chain-adjacent to $u_1$, $v$ be $3$-chain-adjacent to $v_1, v_2$ and $v_3$, and $w$ be $3$-chain-adjacent to $w_1$ and $2$-chain-adjacent to $w_2$. 
It is possible to make two out of the three $3$-chains incident to $v$ even, by pushing  $v$ (if required). 
Similarly, it is possible to make the $3$-chains incident to 
$u$ and $w$ even. 
Notice that, $v_1, v_2$ and $v_3$ allow two colors at $v$, and 
$w_1$ and $w_2$ allow two colors at $w$. That means $v$ and $w$ forbid at most two colors at $u$ while $u_1$ does not forbid any color at $u$.  Therefore, it is possible to extend $f$. Thus, \textbf{C8} is reducible.~$\blacklozenge$

  \medskip

\noindent \textbf{(C9)} Let $u$ be a $4^7$-vertex, $1$-chain-adjacent to a $3^6$-vertex $v$, and  adjacent to a $3^6$-vertex $w$. 
Note that, $u$ must be a $4_{3,3,1,0}$-vertex, 
$v$ must be a $3_{3,2,1}$-vertex,
and 
$w$ must be a $3_{3,3,0}$-vertex. 
Let $u$ be $3$-chain-adjacent to $u_1$ and $u_2$; 
$v$ be $3$-chain-adjacent to $v_1$ and $2$-chain-adjacent to $v_2$; and 
$w$ be $3$-chain-adjacent to $w_1$ and $w_2$. 
It is possible to make the $3$-chain incident to $v$
and one of the two $3$-chains incident to $w$ even, by 
pushing $v,w$ (if required). 
Similarly, it is possible to make two out of the three chains (two $3$-chains and one $1$-chain) incident to $u$ even by pushing $u$ (if required). 
Notice that, 
$v_1$ and $v_2$ allow two colors at $v$, 
and $w_1$ and $w_2$ allow two colors at $w$. 
Therefore, $v$ forbids at most one color at $u$ and 
$w$ forbids no (resp., at most one) colors at $u$ if the $1$-chain connecting $v$ to $u$ is even (resp., odd). Moreover, if the $1$-chain is even (resp., odd), then at least one (resp., two) of the 
chains connecting $u$ with $u_1$ and $u_2$ is even. 
Notice that, 
$u_1$ and $u_2$ forbid at most one (resp., no) color at $u$ if at least one (resp., both) of the chains connecting $u$ with $u_1$ and $u_2$ is even. That means, a total of two colors are forbidden at $u$. Therefore, it is possible to extend $f$. Thus, \textbf{C9} is reducible.~$\blacklozenge$

  \medskip

\noindent \textbf{(C10)} Let $u$ be a $3^4$-vertex, 
$1$-chain-adjacent to a $4^{10}$-vertex $v$ and 
a $3_{3,1,1}$-vertex $w$. 
Note that, $u$ must be a $3_{2,1,1}$-vertex and  $v$ must be a $4_{3,3,3,1}$-vertex. 
It is possible to make two out of the three 
$3$-chains incident to $v$ even, 
the $1$-chain connecting $u$ and $v$ even, and two out of the chains incident to $w$ even by pushing $u,v$ and $w$ (if required). 
Notice that, $v$ does not forbid any color at 
$u$ and $w$ forbids at most one color at $u$. 
Moreover, the vertex $2$-chain-adjacent to $u$ forbids at most one color at $u$. That means, at most two colors are forbidden at $u$. 
 Therefore, it is possible to extend $f$. Thus, \textbf{C10} is reducible.~$\blacklozenge$
  \medskip

\noindent \textbf{(C11)} Let $u$ be a $4^7$-vertex, $1$-chain-adjacent to two $3^6$-vertices $v, w$. 
Note that, $u$ must be a $4_{3,2,1,1}$-vertex while $v, w$ must be  $3_{3,2,1}$-vertices. 
Let $u$ be $3$-chain-adjacent to $u_1$ and $2$-chain-adjacent to $u_2$; $v$ be $3$-chain-adjacent to $v_1$ and $2$-chain-adjacent to $v_2$; and $w$ be $3$-chain-adjacent to $w_1$ and $2$-chain-adjacent to $w_2$. It is possible to make the $3$-chains incident to $u, v$ and $w$ even, by pushing $u,v,w$ (if required). 
In this scenario, $v_1$ and $v_2$ will allow at least two colors at $v$ and $w_1$ and $w_2$ will allow at least two colors at $w$. 
Thus, $v$ and $w$ will forbid at most one color each at $u$ while 
$u_1$ and $u_2$ may forbid one color. Hence there is a chance that all three colors are forbidden at $u$, and $f$ is not readily extendable. 
However, this may happen only  when 
both the $1$-chains connecting $u$ to $v$ and $w$ are odd. In that case, push $u$ to make both of them even. However, this will result in also making the $3$-chain incident to $u$ odd. Thus. $v$ and $w$ will forbid no color at $u$ while $u_1$ and $u_2$ will forbid at most two colors.  Therefore, it is possible to extend $f$. Thus, \textbf{C11} is reducible.~$\blacklozenge$

  \medskip

\noindent \textbf{(C12)}  Let $u$ be a $4^6$-vertex, $1$-chain-adjacent to three $3^6$-vertices $v_1, v_2$ and $v_3$. 
Note that, $u$ must be a $4_{3,1,1,1}$-vertex and $v_1, v_2$ and $v_3$ must be $3_{3,2,1}$-vertices.
It is possible to make the $3$-chains incident to $v_1, v_2, v_3$ even, and two out of the four chains incident to $u$ even. 
Notice that, now at most two colors are forbidden at $u$.
 Therefore, it is possible to extend $f$. 
 Thus, \textbf{C12} is reducible.~$\blacklozenge$

  \medskip

\noindent \textbf{(C13)} 
Let $u$ be a $4^4$-vertex, $1$-chain-adjacent to four $3^6$-vertices $v_1, v_2, v_3$ and $v_4$. 
Note that, $u$ must be a $4_{1,1,1,1}$-vertex and $v_1, v_2, v_3$ and $v_4$ must be $3_{3,2,1}$-vertices.
It is possible to make the $3$-chains incident to $v_1, v_2, v_3$ and $v_4$ even, and two out of the four $2$-chains incident to $v$ even. 
Observe that, now at most two colors are forbidden at $v$.
 Therefore, it is possible to extend $f$. Thus, \textbf{C13} is reducible.~$\blacklozenge$

  \medskip

\noindent \textbf{(C14)}  Notice that it is possible to make the $3$-chain, incident to $x$ even, by pushing (if required) 
the set $\{x,u_1, u_2\}$. Similarly, it is possible to make the $3$-chain, incident to $y$ even, by pushing (if required) the set $\{y,z, u_3\}$. Observe that, $\overrightarrow{C}$ still remains a directed cycle. Now, the chains incident to $\overrightarrow{C}$ do not forbid any color at $x$ and $y$, and allow at least one color at $z$.
Therefore, $f$ is extendable. 
 Thus, \textbf{C14} is reducible.~$\blacklozenge$

\medskip

\noindent \textbf{(C15)} It is enough to prove for the case when $y = 3^2$-vertex and $z=3^4$-vertex. Notice that it is possible to 
make the $3$-chain incident to $x$ even, by pushing (if required) 
the set $\{x,u_1, u_2\}$. Also,
it is possible to make the $1$-chain incident to $y$ (which is not part of $\overrightarrow{C}$) even, by pushing (if required) the set $\{y,z, u_3\}$. Now, $\overrightarrow{C}$ still remains a directed cycle
and the chains incident to $\overrightarrow{C}$ allow all colors at $x$, forbid at most one color at $y$, and forbid at most two colors at $z$, respectively. 
Without loss of generality, assume that the colors $\overline{1},\overline{2}$ are allowed at $y$ and  $\overrightarrow{C}$ is oriented in such a way that $yu_1, u_1z$ are arcs.  In such a scenario, $f$ can be readily extended unless the only color allowed at $z$ is $\overline{2}$. If that happens, then push the set 
 $\{y, z, u_3\}$. Notice that $\overrightarrow{C}$ is still a directed cycle. However, its orientations have changed in such a way that 
 $zu_1, u_1y$ are arcs. 
 Also, now the color $\overline{0}$ is allowed at $y$ and the colors $\overline{0},\overline{1}$ are allowed at $z$. Therefore, it is possible to extend $f$ by assigning $f(y) = \overline{0}$ and $f(z) = \overline{1}$. 
 Thus, \textbf{C15} is reducible.~$\blacklozenge$

\medskip

\noindent \textbf{(C16)}  
We prove for the case when $y = 3^4$-vertex and $z=3^2$-vertex. It is possible to make the $3$-chain, incident to $x$ 
even, by pushing (if required) 
the set $\{x,y,z\}$ keeping $\overrightarrow{C}$ directed. 
We know that the $2$-chain incident to $y$ can forbid at most one color at $y$, and the $3^+$-vertex adjacent to $z$ 
allows at least one color at $z$. 
Choose two distinct colors $i,j$ allowed at $y$ and $z$ respectively, and assign $f(y) = i$ and $f(z) = j$. Without loss of generality, assume that $j = i+2$. Therefore, if $\overrightarrow{C}$ is oriented in such a way that $yu_2, u_2z$ are arcs, then $f$ can be readily extended. If not, push the set 
$\{u_1, u_2, u_3\}$ and  orientation of every arc in 
 $\overrightarrow{C}$ will get reversed. 
 Now it will be possible to extend $f$. 
 Thus, \textbf{C16} is reducible.~$\blacklozenge$

 \medskip

 This completes the proof of the lemma. 
\end{proof}

\subsubsection{Using potential}
\begin{lemma}\label{lem 3-3}
 If $v$ is a $3$-vertex in $\overrightarrow{M}$, adjacent to $x$ and $y$ and  
 $3$-chain-adjacent to $z$, 
 then one of the following is true:
 \begin{enumerate}[(i)]
     \item $x$ and $y$ are adjacent,

     \item there is an oriented $6$-cycle of the form $xw_1w_2w_3yvx$ in $\overrightarrow{M}$ 
 with even number of forward and backward arcs.  
 \end{enumerate} 
\end{lemma}

\begin{proof}
   Let $v_1,v_2$ and $v_3$ be the vertices that are chain incident to $v$ and $z$. Assume that $x$ and $y$ are neither adjacent, nor are a part of an oriented $6$-cycle of the form $xw_1w_2w_3yvx$ in $\overrightarrow{M}$ 
 with even number of forward and backward arcs. 
 Without loss of generality, we may assume that  $vx$ and $vy$ are arcs of $\overrightarrow{M}$ (if not, then push some of the vertices $v, x$ or $y$ accordingly and replace $\overrightarrow{M}$ with its appropriate push equivalent oriented graph). 

 Next, identify the vertices $x$ and $y$ (denote the new identified vertex by $w$)
 to obtain the oriented graph $\overrightarrow{M}'$. This operation gives us a homomorphism $f$ of $\overrightarrow{M}$ to $\overrightarrow{M}'$, where $f(x)=f(y) = w$ and $f(u) = u$ when $u \neq x, y$.  
 Suppose there exists 
 $g: \overrightarrow{M}' - \{v, v_1, v_2, v_3\} \xrightarrow{push} \overrightarrow{C}_3$. 
 If $g$ is a partial homomorphism of $\overrightarrow{M}'$ where $X = \overrightarrow{M}'[\{w, v, v_1, v_2, v_3, z\}]$ is a reducible configuration according to Lemma~\ref{lem reducible-critical}, then $g$ is extendable. 
 Let $g_{ext}$ be the extension of $g$ to a pushable homomorphism of $\overrightarrow{M}'$ to $\overrightarrow{C}_3$. Then $g_{ext} \circ f$ is a 
 pushable homomorphism of 
 $\overrightarrow{M}$ to $\overrightarrow{C}_3$, which contradicts our assumption. Thus, $\overrightarrow{M}' - \{v, v_1, v_2, v_3\}$ does not admit a pushable homomorphism to $\overrightarrow{C}_3$, and hence contains a pushably $\overrightarrow{C}_3$-critical subgraph $\overrightarrow{M}''$.

 Now, $\overrightarrow{M}''$ cannot be 
 $\overrightarrow{C}_{-4}$ since there does not exist any oriented $6$-cycle of the form $xw_1w_2w_3yvx$ in $\overrightarrow{M}$ with even number of forward and backward arcs. 
 Furthermore, notice that $\overrightarrow{M}''$ must contain the vertex $w$, as otherwise, it is a proper subgraph of $\overrightarrow{M}$, and thus cannot be $\overrightarrow{C}_3$-critical.

 Let 
 $Y = (V(\overrightarrow{M}'') \setminus \{w\}) \cup \{v, x, y\}$ and 
 let 
 $\overrightarrow{H}$ be the proper induced subgraph  
 $\overrightarrow{M}[Y]$ of 
 $\overrightarrow{M}$. Let us now calculate the potential of $\overrightarrow{H}$. Since, $\overrightarrow{M}'' \neq \overrightarrow{C}_{-4}$, and as $|V(\overrightarrow{M}'')| + |A(\overrightarrow{M}'')| < |V(\overrightarrow{M})| + |A(\overrightarrow{M})|$, we have $\rho(\overrightarrow{M}'') \leq 0$.   
  But
 $\overrightarrow{H}$ has $2$ vertices and $2$ arcs more than that of $\overrightarrow{M}''$, and we have 
 $\rho(\overrightarrow{H}) = \rho(\overrightarrow{M}'') + (2 \times 15) - (2 \times 13) \leq 4$, a contradiction to Lemma~\ref{lem gap-lemma}. 
\end{proof}

\begin{lemma}\label{lem 3-2} 
    If $v$ is a $3$-vertex in $\overrightarrow{M}$ 
    $1$-chain-adjacent to $x, y$ 
    (with internal vertices $x', y'$, respectively),  
    and $2$-chain-adjacent to $z$ (with internal vertices $v_1, v_2$), then one of the following is true: 
    \begin{enumerate}[(i)]
        \item $x,y$ are adjacent,

        \item there is an oriented $6$-cycle of the form $xw_1yy'vx'x$ in $\overrightarrow{M}$ 
 with even number of forward and backward arcs,

 \item there is an oriented $8$-cycle of the form $xw_1w_2w_3yy'vx'x$ in $\overrightarrow{M}$ 
 with odd number of forward and backward arcs,

\item  there is a pushably isomorphic copy of 
$\overrightarrow{F} = \overrightarrow{E}_i - v^*$ in $\overrightarrow{M} - \{x,x',y,y',v\}$, for some $i \in \{1, 2, 3\}$ and for some $v^* \in V(\overrightarrow{E}_i)$, such that 
both $x$ and $y$ are adjacent to some (not necessarily distinct) vertices of $F$, and moreover, if we identify $x$ and $y$, then together with $V(\overrightarrow{F})$ they will induce $\overrightarrow{E}_i$.
    \end{enumerate}
\end{lemma}

\begin{proof}
Assume that the statement of the lemma is false, and thus $x$ and $y$ are not a part of any of the structures mentioned in the statement through points $(i)-(iv)$. 
 Without loss of generality we may assume that  $vx', x'x, vy',$ and $y'y$ are arcs of $\overrightarrow{M}$ (if not, then push some of $v,x', y', x, y$ accordingly and replace $\overrightarrow{M}$ with its appropriate push equivalent oriented graph). 

 Next identify the vertices $x$ and $y$ (denote the new identified vertex by $w$) and the vertices $x'$ and $y'$ (denote the new identified vertex by $w'$)
 to obtain the oriented graph $\overrightarrow{M}'$. Notice that this operation actually gives us a homomorphism $f$ of $\overrightarrow{M}$ to $\overrightarrow{M}'$, where $f(x)=f(y) = w$, 
 $f(x') = f(y') = w'$ 
 and $f(u) = u$ when $u \neq x, y,x',y'$.  
 Suppose there exists 
 $g: \overrightarrow{M}' - \{w', v, v_1, v_2\} \xrightarrow{push} \overrightarrow{C}_3$. 
 If $g$ is a partial homomorphism of $\overrightarrow{M}'$ where $X = \overrightarrow{M}'[\{w, w', v, v_1, v_2, z\}]$ is a reducible configuration according to Lemma~\ref{lem reducible-critical}, then $g$ is extendable. 
 Let $g_{ext}$ be the extension of $g$ to a pushable homomorphism of $\overrightarrow{M}'$ to $\overrightarrow{C}_3$. Then $g_{ext} \circ f$ is a 
 pushable homomorphism of 
 $\overrightarrow{M}$ to $\overrightarrow{C}_3$, which contradicts our assumption. Thus, we can say that $\overrightarrow{M}' - \{w', v, v_1, v_2\}$ does not admit a pushable homomorphism to $\overrightarrow{C}_3$, and hence contains a pushably $\overrightarrow{C}_3$-critical subgraph $\overrightarrow{M}''$.

 Observe that $\overrightarrow{M}''$ cannot be
 push equivalent to 
 $\overrightarrow{C}_{-4}$ or 
 $\overrightarrow{E}_i$, for any $i \in \{1, 2, 3\}$, since we have assumed that  $x,y$ does not participate in any of the structures described in the points (i)-(iv) in the statement of the lemma. 
 Furthermore notice that, $\overrightarrow{M}''$ must contain the vertex $w$, as otherwise, it is a proper subgraph of $\overrightarrow{M}$, and thus cannot be $\overrightarrow{C}_3$-critical.

 Let 
 $Y = (V(\overrightarrow{M}'') \setminus \{w\}) \cup \{v, x', y', x, y\}$ and 
 let 
 $\overrightarrow{H}$ be the proper induced subgraph  
 $\overrightarrow{M}[Y]$ of 
 $\overrightarrow{M}$. Let us now calculate the potential of $\overrightarrow{H}$. Since $\overrightarrow{M}'' \neq \overrightarrow{C}_{-4}$ and $\overrightarrow{M}'' \not\equiv_p \overrightarrow{E}_i$ for all $i \in \{1, 2, 3\}$, and as $|V(\overrightarrow{M}'')| + |A(\overrightarrow{M}'')| < |V(\overrightarrow{M})| + |A(\overrightarrow{M})|$, we have $\rho(\overrightarrow{M}'') \leq -2$.   
But
 $\overrightarrow{H}$ has $4$ vertices and $4$ arcs more than that of $\overrightarrow{M}''$, and we have 
 $\rho(\overrightarrow{H}) = \rho(\overrightarrow{M}'') + (4 \times 15) - (4 \times 13) \leq 6$, a contradiction to Lemma~\ref{lem gap-lemma}. 
\end{proof}


\begin{lemma}\label{lem Ei+3-2}
The following structure does not exist in $\overrightarrow{M}$: \\

 ``Let $u$ 
be a $3$-vertex $2$-chain adjacent to $u_1$ with internal vertices $u_{11}, u_{12}$, $2$-chain adjacent to $u_2$ with internal vertices $u_{21}, u_{22}$, and $1$-chain adjacent to $u_3$ with internal vertex $u_{31}$, where $u_1$ is not necessarily 
a $3^+$-vertex, $u_2, u_3$ are necessarily $3^+$-vertices, and  $u_2$ is adjacent to $u_3$. 
Let $\overrightarrow{F}$ be pushably isomorphic to 
$\overrightarrow{E}_i - v^*$ for 
some $i \in \{1,2,3\}$ and for some $v^* \in V(\overrightarrow{E}_i)$. 
Let $u_1 \in V(\overrightarrow{F})$ and $u_3$ be adjacent to some vertices (at least one) of $\overrightarrow{F}$.  Moreover, the vertices $u, u_{11}, u_{12}, u_{21}, u_{22}, u_{31}$ and $u_3$ do not belong to $\overrightarrow{F}$ while $u_2$ may or may not belong to 
$\overrightarrow{F}$. 
Finally, if we consider the subgraph induced by $V(\overrightarrow{F})$ together with the vertex obtained by identifying $u_{12}$
and $u_3$, we will get a pushably isomorphic copy of 
$\overrightarrow{E}_i$."  
\end{lemma}

\begin{proof} Let $\overrightarrow{H}$ be the subgraph induced by  $V(\overrightarrow{F}) \cup \{u, u_{11}, u_{12}, u_{21}, u_{22}, u_{31}, u_2, u_3\}$. 
Suppose $u_2 \not\in V(\overrightarrow{F})$. 
Since $|V(\overrightarrow{E}_i)| = 13$ and $|A(\overrightarrow{E}_i)| = 15$, 
we must have $|V(\overrightarrow{F})| = 12$ and 
$|A(\overrightarrow{F})| + [V(\overrightarrow{F}) , \{u_{12}, u_3\}] \geq 15$. 
Observe that, the subgraph induced by $V(\overrightarrow{H}) \setminus V(\overrightarrow{F})$ contains $8$ vertices and $8$ arcs. 
Since, 
$$\rho(\overrightarrow{H}) \leq \rho(\overrightarrow{H} \setminus \overrightarrow{F}) + \rho(\overrightarrow{F}),$$
$$\rho(\overrightarrow{H}) \leq (12+8) \times 15 - (15+8) \times 13 = 1,$$
which implies $\overrightarrow{H} = \overrightarrow{M}$ by Lemma~\ref{lem gap-lemma}.


It is given that $u_2$ and $u_3$ are adjacent and $u_2$  is a $3^+$ vertex. 
Further, except $u_2$, the 
neighborhood of each vertex from 
$V(\overrightarrow{H}) \setminus V(\overrightarrow{F})$  is known. That means, $u_2$ must have a neighbor inside 
$\overrightarrow{F}$.
That gives an extra arc which was not counted in 
$\overrightarrow{H}$ while calculating its potential. Hence, the 
updated calculation will imply
$$\rho(\overrightarrow{M}) = \rho(\overrightarrow{H}) \leq 1 - 13 =-12,$$
a contradiction. This implies $u_2 \in V(\overrightarrow{F})$.

We can again calculate the potential of $\overrightarrow{H}$ supposing 
$u_2 \in V(\overrightarrow{F})$,
similarly like above. 
Observe that since there are exactly 
$7$ vertices and $7$ arcs
in the subgraph induced by $V(\overrightarrow{H}) \setminus V(\overrightarrow{F})$, we have 
$$\rho(\overrightarrow{H}) \leq (12+7) \times 15 - (15+7) \times 13 = -1,$$
which implies $\overrightarrow{H} = \overrightarrow{M}$ by Lemma~\ref{lem gap-lemma}. Notice that if there are any arcs other than the ones we counted, the potential, $\rho(\overrightarrow{M}) = \rho(\overrightarrow{H})$ will be less than or equal to $-14$ which is a contradiction. If not, it is possible to reconstruct $\overrightarrow{M}$ from the information we have by varying 
$i \in \{1,2,3\}$ and $v^* \in V(\overrightarrow{E}_i)$. Note that, since $u_3$ is a $3^{+}$-vertex, it must have a neighbor other than $u_{31}$ and $v_2$ inside $V(\overrightarrow{F})$. Thus, the vertex $v^*$ obtained after merging $u_3$ and $u_{12}$ must be a $3$-vertex of 
$\overrightarrow{E}_i$.

We have observed that all possible $\overrightarrow{M}$ that can be reconstructed in this way 
admit a pushable $3$-coloring. 
This observation is a tedious mechanical task, and we have verified all possible cases using a computer. The detailed code and the results can be found in the following link: \url{https://drive.google.com/drive/folders/1j2wlG9WrvrzqMxaU6pRpLbq1x0FYpMPk?usp=drive_link}. 
\end{proof}


\begin{lemma}\label{lm 5 cycle reduction}
 $\overrightarrow{M}$ does not contain a 
   $5$-cycle $\overrightarrow{C}$ of the form  $vu_1v'u_2u_3v$, where $u_1, u_2, u_3$ are $2$-vertices.
\end{lemma}

\begin{proof}
    Assume the contrary. Suppose, $\overrightarrow{M}$ contains a 
   $5$-cycle $\overrightarrow{C}$ of the form  $vu_1v'u_2u_3v$, where $u_1, u_2, u_3$ are $2$-vertices. We push the vertices of $\overrightarrow{M}$ to obtain an orientation where $vu_1$ and $v'u_1$ are arcs. We get an oriented graph $\overrightarrow{M}_1$ by deleting the vertices $u_1, u_2, u_3$ and identifying the 
    vertices $v, v'$ (call the new vertex $w$). 
    
    If $f$ is a pushable homomorphism of  $\overrightarrow{M}_1$ to $\overrightarrow{C}_3$ with 
    $f(w) = \bar{0}$, without loss of generality, then there exists a partial pushable homomorphism $g$ of $\overrightarrow{M}$ with $g(v) = g(v') = \bar{0}$ where $u_1, u_2, u_3$ are the uncolored vertices. 
    Notice that $g$ is extendable due to Observation~\ref{obs path analysis}, which is a contradiction. 
    
    Hence, $\overrightarrow{M}_1$ does not admit a pushable $3$-coloring, and thus must contain a pushably $3$-critical subgraph $\overrightarrow{M}'$. 
    Now, $\overrightarrow{M}'$ must contain the vertex $w$, as otherwise $\overrightarrow{M}'$ is a 
    pushably $3$-critical proper subgraph of $\overrightarrow{M}$, which is impossible.

    Let $\overrightarrow{M}''$ be the subgraph of $\overrightarrow{M}$ induced by $[V(\overrightarrow{M}') \setminus \{w\}] \cup \{v, u_1, v'\}$. Since $\overrightarrow{M}'$ is a pushably $3$-critical graph having less number of vertices than $\overrightarrow{M}$, 
    unless $\overrightarrow{M}' = \overrightarrow{C}_{-4}$, 
    we must have $\rho(\overrightarrow{M}') \leq 0$. Using this,
    and assuming $\overrightarrow{M}' \neq \overrightarrow{C}_{-4}$,
    we have 
    $$\rho(\overrightarrow{M}'') = \rho(\overrightarrow{M}') + 2 \times 15 - 2\times 13 \leq 4.$$
    However, this implies that $\overrightarrow{M}'' = \overrightarrow{M}$ by Lemma~\ref{lem gap-lemma}, which is impossible since $u_2$ and $u_3$ do not belong to $\overrightarrow{M}''$. 

    That means, $\overrightarrow{M}' = 
    \overrightarrow{C}_{-4}$. Thus, we can reconstruct 
    $\overrightarrow{M}''$ along with the $2$-chain 
    $v'u_2u_3v$, it will exactly be the oriented graph 
    $\overrightarrow{M}'''$ depicted in Fig.~\ref{fig:M''}. 
    Note that, $\rho(\overrightarrow{M}''') = 3$, and thus we 
    must have $\overrightarrow{M}''' = \overrightarrow{M}$. 
    However, this is impossible since $\overrightarrow{M}'''$ 
    is pushably $3$-colorable (see Fig.~\ref{fig:M''}). 
    \end{proof}

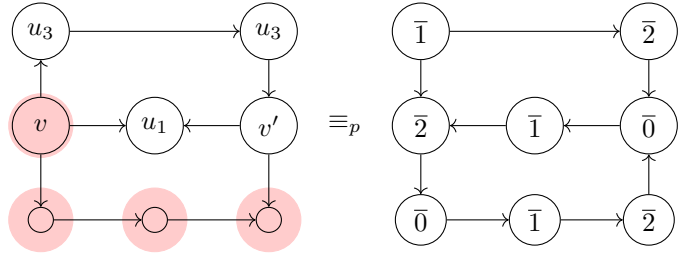
\begin{figure}
    \centering

    \begin{tikzpicture}
    \filldraw[red!20] (0,0) circle (12pt);
    \node[draw, circle, minimum size=.75cm] (v1) at (0,0) {$v$};
    
    \node[draw, circle, minimum size=.75cm] (v2) at (1.5,0) {$u_1$};
    \node[draw, circle, minimum size=.75cm] (v3) at (3,0) {$v'$};
    \node[draw, circle, minimum size=.75cm] (v4) at (0,1.25) {$u_3$};
    \node[draw, circle, minimum size=.75cm] (v5) at (3,1.25) {$u_3$};
    \filldraw[red!20] (0,-1.25) circle (12pt);
    \node[draw, circle, minimum size=.05cm] (v11) at (0,-1.25) {};
    \filldraw[red!20] (1.5,-1.25) circle (12pt);
    \node[draw, circle, minimum size=.05cm] (v21) at (1.5,-1.25) {};
    \filldraw[red!20] (3,-1.25) circle (12pt);
    \node[draw, circle, minimum size=.05cm] (v31) at (3,-1.25) {};
    
     \node at (4,0){$\equiv_p$};

     \node[draw, circle, minimum size=.75cm] (v1') at (5,0) {$\overline{2}$};
    \node[draw, circle, minimum size=.75cm] (v2') at (6.5,0) {$\overline{1}$};
    \node[draw, circle, minimum size=.75cm] (v3') at (8,0) {$\overline{0}$};
    \node[draw, circle, minimum size=.75cm] (v4') at (5,1.25) {$\overline{1}$};
    \node[draw, circle, minimum size=.75cm] (v5') at (8,1.25) {$\overline{2}$};
    \node[draw, circle, minimum size=.05cm] (v11') at (5,-1.25) {$\overline{0}$};
    \node[draw, circle, minimum size=.05cm] (v21') at (6.5,-1.25) {$\overline{1}$};
    \node[draw, circle, minimum size=.05cm] (v31') at (8,-1.25) {$\overline{2}$};
     
    \draw[->] (v1) -- (v2);
    \draw[->] (v3) -- (v2);
     \draw[->] (v1) -- (v4);
     \draw[->] (v5) -- (v3);
     \draw[->] (v4) -- (v5);
     \draw[->] (v1) -- (v11);
     \draw[->] (v11) -- (v21);
     \draw[->] (v21) -- (v31);
     \draw[->] (v3) -- (v31);

     \draw[<-] (v1') -- (v2');
    \draw[->] (v3') -- (v2');
     \draw[<-] (v1') -- (v4');
     \draw[->] (v5') -- (v3');
     \draw[->] (v4') -- (v5');
     \draw[->] (v1') -- (v11');
     \draw[->] (v11') -- (v21');
     \draw[->] (v21') -- (v31');
     \draw[<-] (v3') -- (v31');

\end{tikzpicture}

    \caption{A depiction of the oriented graph $\overrightarrow{M}'''$ as per the proof of Lemma~\ref{lm 5 cycle reduction}, and its pushable $3$-coloring. The highlighted vertices are pushed.}
    \label{fig:M''}
\end{figure}

\begin{lemma}\label{lm 3_221-1-3>4}
    $\overrightarrow{M}$ does not contain a $3_{2,2,1}$-vertex which is $1$-chain-adjacent to a 
    $3^{\geq 4}$-vertex.
\end{lemma}

\begin{proof}
   Suppose $\overrightarrow{M}$ contains a $3_{2,2,1}$-vertex $v$ which is $1$-chain-adjacent to a 
    $3^{\geq 4}$-vertex $u$. Let the two $2$-chains incident to $v$ be 
    $vu_{11}u_{12}v_1$ and $vu_{21}u_{22}v_2$. 
    Also let $u'$ be the $2$-vertex adjacent to $v$ and $u$.  Notice that it is possible to apply Lemma~\ref{lem 3-2} on $v$ 
    in two ways depending on whether $v_1$ or $v_2$ play the role of $z$ (as in Lemma~\ref{lem 3-2}). We arrive at a contradiction by checking all instances stated in Lemma~\ref{lem 3-2}.
\medskip

  \noindent \textbf{Case 1:} Suppose while applying Lemma~\ref{lem 3-2} on $v$, in both instances point $(iv)$ gets implemented. 
    That is, $\overrightarrow{M}$ contains two subgraphs 
    $\overrightarrow{F}_1 = \overrightarrow{E}_i - v_1^*$ and 
    $\overrightarrow{F}_2 = \overrightarrow{E}_j - v_2^*$
    for some $i, j \in \{1,2,3\}$ and for some $v_1^* \in V(\overrightarrow{E}_i)$, $v_2^* \in V(\overrightarrow{E}_j)$. 
    Moreover, $v_2 \in V(\overrightarrow{F}_1)$, $v_1 \in V(\overrightarrow{F}_2)$ and $u, u', v, u_{11}, u_{12}, u_{21}, u_{22} \not\in V(\overrightarrow{F}_1) \cup V(\overrightarrow{F}_2)$. Observe that, the vertices of 
    $\overrightarrow{F}_1$, along with the vertices $u$ and $u_{22}$ identified will give us a copy of $\overrightarrow{E}_i$. Similarly, 
    the vertices of 
    $\overrightarrow{F}_2$, along with the vertices $u$ and $u_{12}$ identified will give us a copy of $\overrightarrow{E}_j$. 
    Let $\overrightarrow{M}'$ be the subgraph of $\overrightarrow{M}$ induced by 
    $V(\overrightarrow{F}_1) \cup V(\overrightarrow{F}_2) \cup \{u, u', v, u_{11}, u_{12}, u_{21}, u_{22}\}$.

    To calculate the potential of $\overrightarrow{M}'$, 
    if we use the value of the potentials of $\overrightarrow{E}_i$ and $\overrightarrow{E}_j$, 
    then we will have counted the potential of two extra vertices (which are not in $\overrightarrow{M}'$). 
    However, we can easily cancel the effect of overcounting two vertices by not counting the vertices $u_{12}$ and $u_{22}$ anymore. 
    On the other hand, some arcs we counted in 
    $\overrightarrow{E}_i$ and $\overrightarrow{E}_j$, may originally have been common arcs between $V(\overrightarrow{F}_1) \cap V(\overrightarrow{F}_2)$ and  $u$.
    There can be at most two such arcs as $u$ is a $3^{\geq 4}$-vertex. 
    Suppose there are 
    $\ell$ arcs between $V(\overrightarrow{F}_1) \cap V(\overrightarrow{F}_2)$ and  $u$. As observed before, $\ell$ may take values 
    $0, 1$ or $2$. 
    That means, the overcounting of $\ell$ arcs can be accounted for by 
    adding a factor of $13 \ell$.
    Furthermore, we have also double counted the potential of the graph induced by $V(\overrightarrow{F}_1) \cap V(\overrightarrow{F}_2)$, and thus we must also subtract a factor of $\rho(\overrightarrow{M}[V(\overrightarrow{F}_1) \cap V(\overrightarrow{F}_2)])$ in our calculation. That leaves us with 
    calculating potential for $5$ more vertices $u, u', v, u_{11}, u_{21}$ and $6$ more arcs, that is, the orientations of the edges 
    $uu'$, $u'v$, $vu_{11}, vu_{21}, u_{11}u_{12}, u_{21}u_{22}$. 
    Therefore, 
    \begin{align*}
    \rho(\overrightarrow{M}') &\leq \rho(\overrightarrow{E}_i) + \rho(\overrightarrow{E}_j) - \rho(\overrightarrow{M}[V(\overrightarrow{F}_1) \cap V(\overrightarrow{F}_2)]) + (5 \times 15) - (6 \times 13) + 13\ell \\
    &\leq 0 + 0 - \rho(\overrightarrow{M}'[V(\overrightarrow{F}_1) \cap V(\overrightarrow{F}_2)]) + -3 + 13\ell \\
    &= 13\ell -3 - \rho(\overrightarrow{M}[V(\overrightarrow{F}_1) \cap V(\overrightarrow{F}_2)]).    
    \end{align*}
 
    Notice that the value of $\rho(\overrightarrow{M}'[V(\overrightarrow{F}_1) \cap V(\overrightarrow{F}_2)])$ is $0$ if the intersection is empty. Otherwise, it is at least $6$ due to Lemma~\ref{lem gap-lemma}. Thus, if $\ell = 0$, we have 
    $\rho(\overrightarrow{M}') \leq -3$, a contradiction. 

    If $\ell = 1$, then $V(\overrightarrow{F}_1) \cap V(\overrightarrow{F}_2)$ is non-empty. Notice that,  $\rho[\overrightarrow{M}'[V(\overrightarrow{F}_1) \cap V(\overrightarrow{F}_2)]) \geq 6$ according to Lemma~\ref{lem gap-lemma}. Thus, 
    $\rho(\overrightarrow{M}') \leq 4$, which implies $\overrightarrow{M}' = \overrightarrow{M}$. 
    That means, the two neighbors of $u$ (other than $u'$), without loss of generality, 
    must both belong to $\overrightarrow{F}_1$. That means, $v_1^*$ is a $3$-vertex of $\overrightarrow{E}_i$. This implies 
$\overrightarrow{M}'[V(\overrightarrow{F}_1) \cap V(\overrightarrow{F}_2)]$ cannot contain all $3$-vertices of $\overrightarrow{E}_i$, and hence cannot have more than one cycle. Therefore, 
$\rho(\overrightarrow{M}'[V(\overrightarrow{F}_1) \cap V(\overrightarrow{F}_2)]) \geq 12$. 
Thus,  
$\rho(\overrightarrow{M}') \leq -2$, a contradiction.

    If $\ell = 2$, then both $v_1^*, v_2^*$ corresponds to a $3$-vertex while 
    $v_1, v_2$, respectively, are one of the neighbors of 
    $v_1^*, v_2^*$. Since $v_1, v_2 \not\in V(\overrightarrow{F}_1) \cap V(\overrightarrow{F}_2)$, the induced subgraph 
    $\overrightarrow{M}'[V(\overrightarrow{F}_1) \cap V(\overrightarrow{F}_2)]$ cannot have more than one cycle. If $\overrightarrow{M}'[V(\overrightarrow{F}_1) \cap V(\overrightarrow{F}_2)]$ has a cycle $C$, then $C$ must have at least six vertices as the girth of each $\overrightarrow{E}_i$ is $6$. Moreover, since $u$ is a $3^{\geq 4}$-vertex, the paths connecting $u$ to $C$ must have at least $3$ vertices outside $C$. Therefore, if 
    $\overrightarrow{M}'[V(\overrightarrow{F}_1) \cap V(\overrightarrow{F}_2)]$ is a connected (resp., disconnected) unicyclic graph, then it must have at least $9$ vertices (resp., $7$) vertices. 
    Hence, if 
    $\overrightarrow{M}'[V(\overrightarrow{F}_1) \cap V(\overrightarrow{F}_2)]$ is a  unicyclic graph, then $\rho(\overrightarrow{M}'[V(\overrightarrow{F}_1) \cap V(\overrightarrow{F}_2)]) \geq 18$, which implies, $\rho(\overrightarrow{M}') \leq 5$, which further implies 
    $\overrightarrow{M}' = \overrightarrow{M}$. 
    This, in particular,  forces 
    $v_1 \not\in V(\overrightarrow{F}_1) \cap V(\overrightarrow{F}_2)$ to be a $3$-vertex of $\overrightarrow{F}_1$, which is not possible since the $C$ must contain all three 
    $3$-vertices of $\overrightarrow{F}_1$. That means, $\overrightarrow{M}'[V(\overrightarrow{F}_1) \cap V(\overrightarrow{F}_2)]$ is a forest. 
    
    If $\overrightarrow{M}'[V(\overrightarrow{F}_1) \cap V(\overrightarrow{F}_2)]$ is a forest with two or more components, then 
$\rho(\overrightarrow{M}'[V(\overrightarrow{F}_1) \cap V(\overrightarrow{F}_2)]) \geq 30$
implying $\rho(\overrightarrow{M}') \leq -7$. a contradiction. That means, $\overrightarrow{M}'[V(\overrightarrow{F}_1) \cap V(\overrightarrow{F}_2)]$ is a tree. Since $\overrightarrow{E}_i$ has girth at least $6$, the 
path connecting the two neighbors of $u$ in 
$V(\overrightarrow{F}_1) \cap V(\overrightarrow{F}_2)$ must have at least $5$ vertices. On the other hand, if $V(\overrightarrow{F}_1) \cap V(\overrightarrow{F}_2)$ has at least $6$ vertices, 
then $\rho(\overrightarrow{M}') \leq -2$, a contradiction.

Hence, we may assume that 
$\overrightarrow{M}'[V(\overrightarrow{F}_1) \cap V(\overrightarrow{F}_2)]$ is a path on $5$ vertices with its endpoints adjacent to $u$. Then, in particular 
$\rho(\overrightarrow{M}') \leq 0$, implying  $\overrightarrow{M}' = \overrightarrow{M}$. Since $u$ is a $3^{\geq 4}$-vertex, and since no 
$3$-vertices of $\overrightarrow{E}_i$ are chain-adjacent through two distinct chains, $V(\overrightarrow{F}_1) \cap V(\overrightarrow{F}_2)$
contains two adjacent $3$-vertices. Moreover, since $v_1^*$ is obtained by identifying $u, u_{12}$, and $u_{12}$ is adjacent to the $3^+$-vertex $v_1$, we must have two adjacent $3$-vertices in $\overrightarrow{E}_i$ which does not belong to $V(\overrightarrow{F}_1) \cap V(\overrightarrow{F}_2)$. That means, $\overrightarrow{E}_i$ contains a pair of non-incident arcs whose endpoints are $3$-vertices. This is a contradiction since $\overrightarrow{E}_i$ does not contain such a pair of edges for any $i \in \{1,2,3\}$.     
    
    Therefore, it is not possible 
 that 
 while applying Lemma~\ref{lem 3-2} on $v$, in both instances point $(iv)$ gets implemented.~$\blacklozenge$
\medskip

  \noindent \textbf{Case 2:} Suppose that while applying Lemma~\ref{lem 3-2} on $v$, with $v_2$ 
playing the role of $z$ (from Lemma~\ref{lem 3-2}), point $(i)$ gets implemented. That will force $u = v_1$, which in turn implies a $5$-cycle of the form $uu'vu_{11}u_{12}u$ where $u', u_{11}, u_{12}$ are $2$-vertices. This is not possible due to Lemma~\ref{lm 5 cycle reduction}. Hence, while applying Lemma~\ref{lem 3-2} on $v$, point $(i)$ of the lemma can not be implemented.~$\blacklozenge$ 
\medskip

  \noindent \textbf{Case 3:} Suppose, applying Lemma~\ref{lem 3-2} on $v$, with $v_2$ 
playing the role of $z$ (from Lemma~\ref{lem 3-2}), point $(ii)$ gets implemented. This forces $u$ and $v_1$ to be adjacent. Moreover, since $u$ is a 
$3^{\geq 4}$-vertex, a $3$-chain $uu_1u_2u_3u''$ 
must be incident to $u$. That means, 
if we apply Lemma~\ref{lem 3-2}
on $v$ with $v_1$ playing the role of $z$ 
(from Lemma~\ref{lem 3-2}) we cannot implement points 
(i)-(iii). Furthermore, it is not possible to impement point (iv) due to Lemma~\ref{lem Ei+3-2}. Hence, while applying Lemma~\ref{lem 3-2} on $v$, point $(ii)$ of Lemma~\ref{lem 3-2} can not be implemented.~$\blacklozenge$
\medskip

  \noindent \textbf{Case 4:} Suppose while applying Lemma~\ref{lem 3-2} on $v$,  the points (iii) and (iv) get implemented, respectively. Without loss of generality, this will force 
a path of the form $uw_{11}w_{12}v_1$ and $\overrightarrow{F} = \overrightarrow{E}_i - v^*$ for some 
$i \in \{1, 2, 3\}$
and some $v^* \in V(\overrightarrow{E}_i)$ containing $v_2$ and having some neighbors of $u$. We know that $\{u$, $v$, $u'$, $u_{11}$, $u_{12}$, $u_{21}$, $u_{22}\}$ $\not\in V(\overrightarrow{F})$. 
Let $\overrightarrow{M}'$ be the subgraph of $\overrightarrow{M}$ induced by $V(\overrightarrow{F})$ $\cup$ $\{u$, $v$, $u'$, $u_{11}$, $u_{12}$, $u_{21}$, $u_{22}$, $v_1$, $w_{11}$, $w_{12}\}$. 
Thus,
$\rho(\overrightarrow{M}') \leq 5$ if $v_1, w_{11}$ and $w_{22}$ do not belong to $\overrightarrow{F}$. However, this will imply $\overrightarrow{M}' = \overrightarrow{M}$. This is impossible since $v_1$ is a $2$-vertex in $\overrightarrow{M}'$
and a $3^+$-vertex in $\overrightarrow{M}$. On the other hand, if we assume that any of $v_1, w_{11}$ or $w_{12}$ belong to $\overrightarrow{F}$, then the potential calculation will give us
$\rho(\overrightarrow{M}') \leq -10$, a contradiction.~$\blacklozenge$ 
\medskip

  \noindent \textbf{Case 5:} Suppose while applying Lemma~\ref{lem 3-2} on $v$, in both instances point $(iii)$ gets implemented. This will force 
two paths of the form $uw_{11}w_{12}v_1$ and $uw_{21}w_{22}v_2$. Also, since $u$ is a $3^{\geq 4}$-vertex, without loss of generality we may assume that $w_{11}, w_{12}, w_{21}$ are $2$-vertices.  Suppose without loss of generality that the arcs $vu'$ and $uu'$ are present in $\overrightarrow{M}$. Since $\overrightarrow{M} - u'$ is pushably $3$-colorable, the pushable $3$-coloring of $\overrightarrow{M} - u'$ can always be extended to $\overrightarrow{M}$.~$\blacklozenge$
\end{proof}

\begin{lemma}\label{lem 3_222}
$\overrightarrow{M}$ does not contain a $3_{2,2,2}$-vertex.
\end{lemma}

\begin{proof}
Suppose $\overrightarrow{M}$ contains a $3_{2,2,2}$-vertex $v$.
     Let the three $2$-chains incident to $v$ be 
    $vu_{11}u_{12}v_1$, $vu_{21}u_{22}v_2$, and $vu_{31}u_{32}v_3$. 
    Notice that it is possible to apply Lemma~\ref{lem 3-2} on $v$ in three ways depending on whether $v_1$, $v_2$ or $v_3$ play the role of $z$ (as in Lemma~\ref{lem 3-2}).  
    \medskip

\noindent \textbf{Case 1:} It is not possible to apply  Lemma~\ref{lem 3-2}$(i)$ on $v$. Hence this case does not arise. ~$\blacklozenge$

\medskip  

\noindent \textbf{Case 2:} Without loss of generality, suppose that Lemma~\ref{lem 3-2}$(ii)$ gets implemented with $v_1$  playing the role of $z$ (as in Lemma~\ref{lem 3-2}). This implies $v_2 = v_3$. 
Let $\overrightarrow{M}'$
be the oriented graph obtained by identifying the vertices $v_3$  and $v_2$ (call this vertex as $w$), 
and deleting the vertices $u_{21}, u_{22}, u_{31}, u_{32}$ from $\overrightarrow{M}$. 
If $\overrightarrow{M}'$ is pushably $\overrightarrow{C}_3$-colorable, then we can extend the coloring to a 
pushable $\overrightarrow{C}_3$-coloring of $\overrightarrow{M}$ due to Observation~\ref{obs path analysis}. That means, 
$\overrightarrow{M}'$ contains a pushably $3$-critical graph $\overrightarrow{M}''$, and due to our assumptions, $\rho(\overrightarrow{M}'') \leq 8$ 
(in case if it is $\overrightarrow{C}_{-4}$). 
Let $Y = (V(\overrightarrow{M}'') \setminus \{w\}) \cup \{u_{21}, u_{22}\}$. 
Consider the proper induced subgraph 
$\overrightarrow{H} = \overrightarrow{M}[Y]$ of $\overrightarrow{M}$. 
Note that, $\rho(\overrightarrow{H}) \leq 8 + (15 \times 2) - (13 \times 3) = -1$. This is a contradiction to Lemma~\ref{lem gap-lemma}. Therefore, it is not possible to implement the point $(ii)$ of Lemma~\ref{lem 3-2} at all on $v$. ~$\blacklozenge$

\medskip

  \noindent \textbf{Case 3:} Suppose that while applying Lemma~\ref{lem 3-2} on $v$, in two instances Lemma~\ref{lem 3-2}$(iv)$ gets implemented. 
Without loss of generality, suppose that Lemma~\ref{lem 3-2}$(iv)$ gets implemented with $v_1$ and $v_2$ playing the role of $z$ (as in Lemma~\ref{lem 3-2}) respectively. 
That means, $\overrightarrow{M}$ contains two subgraphs 
    $\overrightarrow{F}_1 = \overrightarrow{E}_i - v_1^*$ and 
    $\overrightarrow{F}_2 = \overrightarrow{E}_j - v_2^*$
    for some $i, j \in \{1,2,3\}$ and for some $v_1^* \in V(\overrightarrow{E}_i)$, $v_2^* \in V(\overrightarrow{E}_j)$. 
    Moreover, $v_2, v_3 \in V(\overrightarrow{F}_1)$, $v_1, v_3 \in V(\overrightarrow{F}_2)$ and $v, u_{31}, u_{32}, u_{11}, u_{12}, u_{21}, u_{22} \not\in V(\overrightarrow{F}_1) \cup V(\overrightarrow{F}_2)$. Observe that, the vertices of 
    $\overrightarrow{F}_1$, along with the vertices $u_{32}$ and $u_{22}$ identified will give us a copy of $\overrightarrow{E}_i$. Similarly, 
    the vertices of 
    $\overrightarrow{F}_2$, along with the vertices $u_{32}$ and $u_{12}$ identified will give us a copy of $\overrightarrow{E}_j$. 
    Let $\overrightarrow{M}'$ be the subgraph of $\overrightarrow{M}$ induced by 
    $V(\overrightarrow{F}_1) \cup V(\overrightarrow{F}_2) \cup \{v, u_{31}, u_{32},  u_{11}, u_{12}, u_{21}, u_{22}\}$.

    To calculate the potential of $\overrightarrow{M}'$, 
    on the one hand, if we use the value of the potentials of $\overrightarrow{E}_i$ and $\overrightarrow{E}_j$, 
    then we will have counted the potential of two extra vertices (which are not in $\overrightarrow{M}'$). 
    However, we can easily cancel the effect of overcounting two vertices by not counting the vertices $u_{12}$ and $u_{22}$ anymore. 
    Observe that, the arc connecting $u_{32}$ and $v_3$
    is counted two times - once  in $\overrightarrow{E}_i$ and 
    once in $\overrightarrow{E}_j$. To neutralize the effect of this double count, we must add a factor of $13$ in our calculation. 
    Furthermore, we have also double counted the potential of the graph induced by $V(\overrightarrow{F}_1) \cap V(\overrightarrow{F}_2)$, and thus we must also subtract a factor of $\rho(\overrightarrow{M}[V(\overrightarrow{F}_1) \cap V(\overrightarrow{F}_2)])$ in our calculation. That leaves us with 
    calculating potential for $5$ more vertices $v, u_{31}, u_{32}, u_{11}, u_{21}$ and $6$ more arcs, that is, the orientations of the edges 
    $u_{31}u_{32}$, $u_{31}v$, $vu_{11}, vu_{21}, u_{11}u_{12}, u_{21}u_{22}$. 
    Therefore, 
    \begin{align*}
    \rho(\overrightarrow{M}') &\leq \rho(\overrightarrow{E}_i) + \rho(\overrightarrow{E}_j) - \rho(\overrightarrow{M}[V(\overrightarrow{F}_1) \cap V(\overrightarrow{F}_2)]) + (5 \times 15) - (6 \times 13) + 13\\
    &\leq 0 + 0 - \rho(\overrightarrow{M}'[V(\overrightarrow{F}_1) \cap V(\overrightarrow{F}_2)]) + (-3) + 13 \\
    &= 10 - \rho(\overrightarrow{M}[V(\overrightarrow{F}_1) \cap V(\overrightarrow{F}_2)]).    
    \end{align*}
 
    If $\overrightarrow{M}'[V(\overrightarrow{F}_1) \cap V(\overrightarrow{F}_2)]$  has at most one cycle (a cycle must have at least $6$ vertices since the girth of $\overrightarrow{E}_i$ 
    and $\overrightarrow{E}_j$ is $6$), or is a disconnected graph, then 
    $\rho(\overrightarrow{M}'[V(\overrightarrow{F}_1) \cap V(\overrightarrow{F}_2)]) \geq 12$. This implies
     $\rho(\overrightarrow{M}') \leq -2$, a contradiction to Lemma~\ref{lem gap-lemma}.

      Hence, $\overrightarrow{M}'[V(\overrightarrow{F}_1) \cap V(\overrightarrow{F}_2)]$  is a connected graph having at least two cycles. Then it must be the same as $\overrightarrow{E}_i - v_1^*$
      where $v_1^*$ is a $2$-vertex of $\overrightarrow{E}_i$. In that case,      
    $\rho(\overrightarrow{M}'[V(\overrightarrow{F}_1) \cap V(\overrightarrow{F}_2)]) \geq 12$. This implies
     $\rho(\overrightarrow{M}') \leq -2$, a contradiction to Lemma~\ref{lem gap-lemma}. 

     Therefore,  it is not possible 
 that 
 while applying Lemma~\ref{lem 3-2} on $v$, in two instances Lemma~\ref{lem 3-2}$(iv)$ gets implemented.~$\blacklozenge$

\medskip

  \noindent \textbf{Case 4:} Suppose that while applying Lemma~\ref{lem 3-2} on $v$, in one of the instances Lemma~\ref{lem 3-2}$(iv)$ gets implemented. 
  Without loss of generality, suppose that Lemma~\ref{lem 3-2}$(iv)$ gets implemented with $v_1$  playing the role of $z$ (as in Lemma~\ref{lem 3-2}). 
  That means, $\overrightarrow{M}$ contains a subgraph 
    $\overrightarrow{F} = \overrightarrow{E}_i - v^*$ 
    for some $i \in \{1,2,3\}$ and for some $v^* \in V(\overrightarrow{E}_i)$. 
    Moreover, $v_2, v_3 \in V(\overrightarrow{F})$  and 
    $v_1, u_{11}, u_{12}, v, u_{21}, u_{22}, u_{31}, u_{32} \not\in V(\overrightarrow{F})$. 
    Due to the previous cases, the only option for us is to implement Lemma~\ref{lem 3-2}$(iii)$ on $v$ when the role of $z$ (as in Lemma~\ref{lem 3-2}) is played by $v_2$ and $v_3$, respectively.  In this case, $v_1$ must be connected to $v_2$  
  (resp., $v_3$)  by an (underlying)  $2$-path of the form  $v_1w_2v_2$
  (resp., $v_1w_3v_3$). Notice that, the vertices $w_2, w_3$ may or may not belong to $\overrightarrow{F}$. Let $\overrightarrow{M}'$ be the graph induced by $V(\overrightarrow{F}) \cup \{v_1, u_{11}, u_{12}, v, u_{21}, u_{22}, u_{31}, u_{32}, w_2, w_3\}$. Observe that, $\rho(\overrightarrow{M}') \leq -4$, a contradiction. 
  $\blacklozenge$

\medskip

  \noindent \textbf{Case 5:} By the above cases, we are forced to implement Lemma~\ref{lem 3-2}$(iii)$ for all three instances. This will force $\overrightarrow{E}_2$ as a subgraph of 
  $\overrightarrow{M}$, a contradiction.  $\blacklozenge$
\end{proof}

Applying Lemmas~\ref{lem 3-3} and \ref{lem 3-2}, we can reduce some additional configurations as given in the following lemma. 

\begin{figure}[]
    \centering

    \begin{tabularx}{\textwidth}{X X X}
				
\scalebox{.7}{
\begin{tikzpicture}
  
  \node[fill=black!20,circle, minimum size=.75cm,draw] (a1) at (0,0) {$3^6$};
  \node[fill=black!20,circle, minimum size=.75cm,draw] (a3) at (2,0) {$3^{\geq 1}$};
  
  \node at (1,-1) {\textbf{P1}};
  
  \draw[thick] (a1)-- (a3);

  \node[fill=black!20,circle, minimum size=.75cm,draw] (b1) at (4,0) {$3^5$};
  \node[circle, minimum size=.05cm,draw] (b2) at (5.25,0) {};
  \node[fill=black!20,circle, minimum size=.75cm,draw] (b3) at (6.5,0) {$3^5$};
  \node[circle, minimum size=.05cm,draw] (b4) at (7.75,0) {};
  \node[fill=black!20,circle, minimum size=.75cm,draw] (b5) at (9,0) {$3^5$};

  \node at (6.5,-1) {\textbf{P2}};

  \draw[thick] (b1)-- (b2)-- (b3)-- (b4)-- (b5); 
    
\end{tikzpicture}
}

&
\scalebox{.7}{
\begin{tikzpicture}
\node at (0,0) {};
  \node[fill=black!20,circle,minimum size=.75cm,draw] (a1) at (4,0) {$3^6$};
  \node[fill=black!20,circle,minimum size=.75cm,draw] (a2) at (6,0) {$3^0$};
  \node[fill=black!20,circle,minimum size=.75cm,draw] (a3) at (8,0) {$3^6$};

  \node at (6,-1) {\textbf{P3}};

  \draw[thick] (a1)-- (a2)-- (a3);  
 
  \node[fill=black!20,circle,minimum size=.75cm,draw] (b1) at (10,0) {$3^5$};
  \node[fill=black!20,circle,minimum size=.75cm,draw] (b2) at (12,0) {$3^2$};
  \node[fill=white,circle,minimum size=.05cm,draw] (b3) at (13.25,0) {};
  \node[fill=black!20,circle,minimum size=.75cm,draw] (b4) at (14.5,0) {$3^6$};
 
  \node at (12,-1) {\textbf{P4}};
 
  \draw[thick] (b1)-- (b2)-- (b3)-- (b4); 
    
\end{tikzpicture}
}
\\
\scalebox{.7}{
\begin{tikzpicture}
   
  \node[fill=black!20,circle,minimum size=.75cm,draw] (a1) at (0,0) {$3^5$};
  \node[fill=white,circle,minimum size=.05cm,draw] (a2) at (1.25,0) {};
  \node[fill=black!20,circle,minimum size=.75cm,draw] (a3) at (2.5,0) {$3^4$};
  \node[fill=white,circle,minimum size=.05cm,draw] (a4) at (3.75,0) {};
  \node[fill=black!20,circle,minimum size=.75cm,draw] (a5) at (5,0) {$3^5$};

  \node at (2.5,-1) {\textbf{P5}};

  \draw[thick] (a1)-- (a2)-- (a3)-- (a4)-- (a5); 
    
\end{tikzpicture}
}
&
\scalebox{.7}{
\begin{tikzpicture}
   
    \centering
  \node[fill=black!20,circle,minimum size=.75cm,draw] (a1) at (0,0) {$3^5$};
  \node[fill=white,circle,minimum size=.05cm,draw] (a2) at (1.25,0) {};
  \node[fill=black!20,circle,minimum size=.75cm,draw] (a3) at (2.5,0) {$3^3$};
  \node[fill=white,circle,minimum size=.05cm,draw] (a4) at (3.75,0) {};
  \node[fill=black!20,circle,minimum size=.75cm,draw] (a5) at (5,0) {$3^5$};
  \node[fill=white,circle,minimum size=.05cm,draw] (a6) at (2.5,1) {};
  \node[fill=black!20,circle,minimum size=.75cm,draw] (a7) at (2.5,2) {$3^5$};

  \node at (2.5,-1) {\textbf{P6}};

  \draw[thick] (a1)-- (a2)-- (a3)-- (a4)-- (a5); 
  \draw[thick] (a3)-- (a6)-- (a7);
    
\end{tikzpicture}

}

&
\scalebox{.7}{
\begin{tikzpicture}
  \node[fill=black!20,circle,draw] (a1) at (-.5,0) {$3^5$};
  \node[fill=black!20,circle,draw] (a2) at (1.5,0) {$3^1$};
  \node[fill=black!20,circle,draw] (a3) at (3.5,0) {$3^5$};
  \node[fill=white,circle,draw] (a4) at (1.5,1) {};
  \node[fill=black!20,circle,draw] (a5) at (1.5,2) {$3^6$};
 \node at (-1,0) {};
  
  \node at (1.5,-1) {\textbf{P7}};

  \draw[thick] (a1)-- (a2)-- (a3); 
  \draw[thick] (a2)-- (a4)-- (a5); 
    
\end{tikzpicture}
}

\\
\scalebox{.7}{
\begin{tikzpicture}
  \node[fill=black!20,circle,draw] (a1) at (0,0) {$3^5$};
  \node[fill=black!20,circle,draw] (a2) at (2,0) {$3^4$};
  \node[fill=white,circle,draw] (a3) at (3.25,0) {};
  \node[fill=black!20, circle, draw] (a4) at (4.5,0)  {\scriptsize$3^{\geq 5}$}; 
  
  \node at (2,-1) {\textbf{P8}};
  \node at (-1,0) {};

  \draw[thick] (a1) -- (a2)-- (a3)-- (a4); 
    
\end{tikzpicture}
}
&
\scalebox{.7}{
\begin{tikzpicture}
  \node[fill=black!20,circle,draw] (a1) at (0,0) {$3^6$};
  \node[fill=white,circle,draw] (a2) at (1.25,0) {};
  \node[fill=black!20,circle,draw] (a3) at (2.5,0) {$3^{\geq 3}$};
  
  \node at (1.25,-1) {\textbf{P9}};
 \node at (-1.5,0) {};
  \draw[thick] (a1)-- (a2)-- (a3); 
    
\end{tikzpicture}

}
&

\scalebox{.7}{
\begin{tikzpicture}
\centering
  \node[fill=black!20,circle,draw] (a1) at (0,0) {$3^6$};
  \node[fill=white,circle,draw] (a2) at (1.25,0) {};
  \node[fill=black!20,circle,draw] (a3) at (2.5,0) {$3^2$};
  \node[fill=white,circle,draw] (a4) at (3.75,0) {};
  \node[fill=black!20,circle,draw] (a5) at (5,0) {$3^6$};
 
  \node at (2.5,-1) {\textbf{P10}};

  \draw[thick] (a1)-- (a2)-- (a3)-- (a4)-- (a5); 
    
\end{tikzpicture}
}

\end{tabularx}

\caption{The list of reducible configurations (using potential) for Lemma~\ref{lem reducible-potential}.}\label{fig:Pi}
\end{figure}
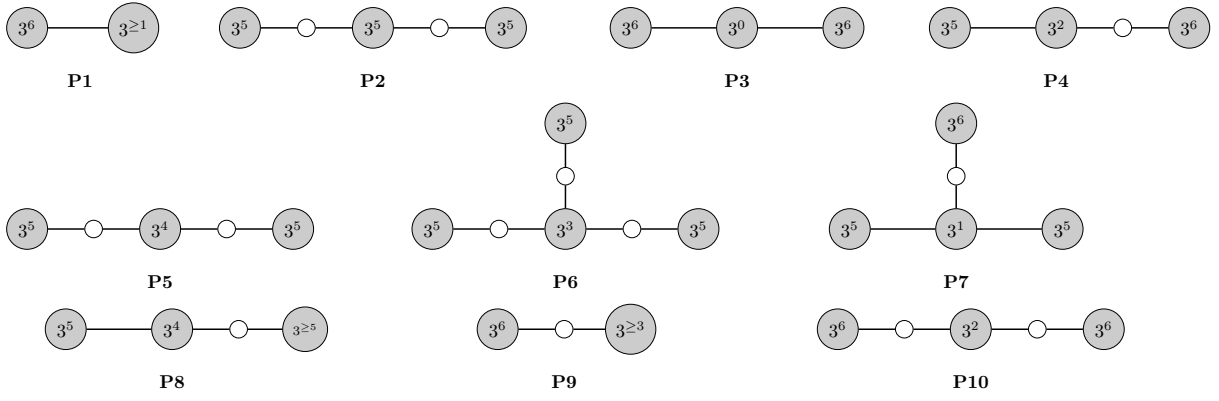

\begin{lemma}\label{lem reducible-potential}
    The configurations listed below are reducible. See Fig.~\ref{fig:Pi} for a pictorial reference. 
\begin{enumerate}[(P1)]
  \item A $3^6$-vertex adjacent to a 
    $3^{\geq 1}$-vertex.

    \item A $3^5$-vertex $1$-chain-adjacent to two 
    $3^5$-vertices.   
   
 \item A $3^0$-vertex adjacent to two  $3^6$-vertices.

\item A $3^2$-vertex adjacent to a $3^5$-vertex and $1$-chain-adjacent to a $3^{6}$-vertex.

\item A $3^4$-vertex $1$-chain-adjacent to two $3^5$-vertices.

\item A $3^{3}$-vertex $1$-chain-adjacent to three 
$3^5$-vertices.

\item  A $3^{1}$-vertex adjacent to two $3^{5}$-vertices and 
$1$-chain-adjacent to a $3^6$-vertex.

\item  A $3^4$-vertex adjacent to a $3^5$-vertex and $1$-chain-adjacent to a $3^{\geq 5}$-vertex.

 \item A $3^6$-vertex $1$-chain-adjacent to a 
    $3^{\geq 3}$-vertex. 

\item A $3^2$-vertex $1$-chain-adjacent to two  $3^6$-vertices. 
\end{enumerate}
\end{lemma}

\begin{proof}
\noindent   \textbf{(P1)} Let $v$ be a $3^6$-vertex adjacent to the $3^{\geq 1}$-vertex $u$, and $3$-chain-adjacent to $v_1$ and $v_2$. Applying Lemma~\ref{lem 3-3} we can infer that both $v_1$ and $v_2$ are adjacent to $u$. Thus, we get $deg(u) = 3$, which means $u$ must be a $3^0$ vertex, which is impossible. Thus, \textbf{P1} is reducible.~$\blacklozenge$

\medskip

  \noindent   \textbf{(P2)} Let $v_2$ be a $3^5$-vertex $1$-chain-adjacent to two other $3^5$-vertices $v_1$ and $v_3$. 
  Since $v_2$ is a $3^5$-vertex and $v_2$ is incident to a $3$-chain, 
  by applying Lemma~\ref{lem 3-3} we are forced to have an oriented $6$-cycle $\overrightarrow{C}$ of the form $v_1u_1v_2u_2v_3u_3v_1$ where $v_1, v_2$ and $v_3$ are $3^5$-vertices and $u_1, u_2$ and $u_3$ are $2$-vertices. 
  Moreover, note that, each $v_i$ must be $3$-chain 
  adjacent to $v'_i$ (say), and the vertices of the $3$-chains (except $v_i$) are outside the cycle $\overrightarrow{C}$. This is exactly the configuration \textbf{C16} which is reducible due to Lemma~\ref{lem reducible-critical}. Thus, \textbf{P2} is reducible.~$\blacklozenge$

\medskip

\noindent   \textbf{(P3)} Let $v$ be a $3^0$-vertex adjacent to 
  two $3^6$ vertices $v_1$ and $v_2$ (say), let the other neighbor of $v$, a  
  $3^+$-vertex be $v'$. 
  Let $v_i$ be $3$-chain-adjacent to 
  $v_{i1}$ and $v_{i2}$ for all $i \in \{1,2\}$. 
  Thus, applying Lemma~\ref{lem 3-3} on $v_1$ (resp., $v_2$) we can infer that the vertex $v_{1j}$ (resp., $v_{2j}$) must be the same as either $v'$ 
  or $v_2$ (resp., $v_1$), 
  where $j \in \{1,2\}$. 
  However, irrespective of which vertex $v_{ij}$'s get identified with, 
  the vertex $v'$ will be a cut vertex inside $\overrightarrow{M}$, or $\overrightarrow{M}$ is push equivalent to $\overrightarrow{E}_1$ (this is possible when $v'$ is also a $3^6$-vertex). 
  However, since 
  $\overrightarrow{M}$ is pushably $3$-critical, it cannot contain a cut vertex. Furthermore, by our assumption, $\overrightarrow{M}$ cannot be push equivalent to $\overrightarrow{E}_1$. This is a contradiction. 
  Thus, \textbf{P3} is reducible.~$\blacklozenge$

   \medskip

  \noindent   \textbf{(P4)} Let $v$ be a $3^2$-vertex  
  adjacent to a $3^5$-vertex $v_1$,
  $1$-chain-adjacent to a $3^6$-vertex $v_2$, 
and $1$-chain-adjacent to another $3^+$-vertex $v'$. 
  Applying Lemma~\ref{lem 3-3} on $v_1$ and $v_2$ we can infer the existence of a directed $6$-cycle of the form 
  $\overrightarrow{C} = v_1vu_1v_2u_2u_3v_1$, where $u_1, u_2$ and $u_3$ are $2$-vertices. This is exactly the configuration \textbf{C15} which is reducible due to Lemma~\ref{lem reducible-critical}. Thus, \textbf{P4} is reducible.~$\blacklozenge$

\medskip

  \noindent   \textbf{(P5)}  Let $v$ be a $3^4$-vertex  
  $1$-chain-adjacent to two $3^5$-vertices $v_1, v_2$ and
  $2$-chain-adjacent to a vertex $v_3$. Notice that $v$ must be $3_{2,1,1}$-vertex. We know that $v_1$ and $v_2$ are not adjacent by Lemma~\ref{lem reducible-critical} and are not $3_{2,2,1}$ by Lemma~\ref{lm 3_221-1-3>4}. Hence, the vertices $v_1$ and $v_2$ must be $3_{3,1,1}$. Applying Lemma~\ref{lem 3-3} on $v_1$ and $v_2$ we can 
  say that we must have a directed $6$-cycle involving the vertices $v, v_1$ and $v_2$ of the form $\overrightarrow{C} = v_1u_1v_2u_2vu_3v_1$, where $u_1, u_2$ and $u_3$ are $2$-vertices. That is, the configuration reduces to \textbf{C16}, which is reducible due to Lemma~\ref{lem reducible-critical}. Thus, \textbf{P5} is reducible.~$\blacklozenge$

  \medskip

  \noindent   \textbf{(P6)}  Let $v$ be a $3^3$-vertex  
  $1$-chain-adjacent to three $3^5$-vertices $v_1, v_2,$ and $v_3$. Notice that $v$ must be $3_{1,1,1}$-vertex. Suppose $v_1$ is a $3_{2,2,1}$-vertex. If we apply Lemma~\ref{lem 3-2} on $v_1$, then observe that it is not possible to implement points $(i), (ii)$, or $(iii)$. Thus, point $(iv)$ will get implemented twice. This will imply a subgraph with potential $-8$ or less, a contradiction. Hence the vertices $v_1, v_2$ and $v_3$ must be $3_{3,1,1}$. Applying Lemma~\ref{lem 3-3} on $v_1$ and $v_2$ we must have a directed $6$-cycle involving the vertices $v, v_1$ and $v_2$ of the form $\overrightarrow{C} = v_1u_1v_2u_2vu_3v_1$, where $u_1, u_2$ and $u_3$ are $2$-vertices. That is, the configuration reduces to \textbf{C16}, which is reducible due to Lemma~\ref{lem reducible-critical}. Thus, \textbf{P6} is reducible.~$\blacklozenge$

   \medskip

 \noindent   \textbf{(P7)} Let $v$ be a $3^1$-vertex  
  adjacent to two $3^5$-vertices $v_1$ and $v_2$, and
  $1$-chain-adjacent to a $3^6$-vertex $v_3$.
  Notice that $v_1$ and $v_2$ must be $3_{3,2,0}$-vertices and $v_3$ must be a 
  $3_{3,2,1}$-vertex. 
  Thus, applying Lemma~\ref{lem 3-3} on $v_1$ and $v_2$ we can 
  say that we must have a $6$-cycle involving the vertices $v, v_1$ and $v_3$ along with three other $2$-vertices and a $6$-cycle involving the vertices $v, v_2$ and $v_3$ along with three other $2$-vertices. Observe that it is impossible to have such a configuration.  
  Thus, \textbf{P7} is reducible.~$\blacklozenge$

 \medskip

    \noindent   \textbf{(P8)} Let $v$ be a $3^4$-vertex adjacent to a $3^5$-vertex $v_1$, $1$-chain-adjacent to a $3^{\geq5}$-vertex $v_2$, and $3$-chain adjacent to a vertex $v_3$. Notice that $v$ must be a $3_{3,1,0}$-vertex, $v_1$ must be a $3_{3,2,0}$-vertex, and $v_2$ must be a $3_{3,1,1}$-vertex or a $3_{3,2,1}$-vertex. Let $v_1$ be $3$-chain-adjacent to $v_4$ and $v_2$ be $3$-chain-adjacent to $v_5$. Applying Lemma~\ref{lem 3-3} on $v_1$ we must have a directed $6$-cycle involving the vertices $v, v_1$ and $v_2$ of the form $\overrightarrow{C} =vv_1u_1u_2v_2u_3v$, where $u_1, u_2$ and $u_3$ are $2$-vertices. This is exactly the configuration \textbf{C15} which is reducible due to Lemma~\ref{lem reducible-critical}. Thus, \textbf{P8} is reducible.~$\blacklozenge$
    
    \medskip

\noindent   \textbf{(P9)} It is enough to show that a $3^6$-vertex $v$ that is $1$-chain-adjacent to a $3^3$-vertex $v_3$ is reducible. Firstly, $v$ must be $3_{3,2,1}$-vertex. 
Let $v$ 
be incident to the $3$-chain $vu_{11}u_{12}u_{13}v_1$, 
the $2$-chain $vu_{21}u_{22}v_2$, and the $1$-chain $vu_{31}v_3$. Applying Lemma~\ref{lem 3-3} on $v$ forces a directed $6$-cycle of the form $vu_{21}u_{22}v_2v_3u_{31}v$. Next, we can apply Lemma~\ref{lem 3-2} on $v$ and that will force one of the following four scenarios:
\begin{enumerate}[(i)]
    \item The vertices $u_{12}$ and $v_3$ are adjacent. This forces $u_{13} = v_3$ which is impossible since they have different degrees. 

    \item The vertex $u_{12}$ and $v_3$ have a common neighbor. This forces $v_1 = v_3$. Notice that the subgraph $\overrightarrow{M}'$ (say) induced by $\{v, v_1, v_2, u_{11}, u_{12}, u_{13}, u_{21}, u_{22}, u_{31}\}$ 
    has potential $\rho(\overrightarrow{M'}) = 5$. Hence, Lemma~\ref{lem gap-lemma} implies that $\overrightarrow{M}' = \overrightarrow{M}$. However, $v_2$ is a $3^+$ vertex in $\overrightarrow{M}$ but a $2$-vertex in $\overrightarrow{M}'$, a contradiction.

     \item The vertex $u_{12}$ and $v_3$ are connected by a path with three internal vertices $u_{13}, v_1$ and $w$ (say).  Notice that, in this case, even if $w$ is a $2$-vertex, $v_3$ 
     must be a $3^{\leq 2}$-vertex, a contradiction. 

     \item The fourth case forces the situation described in 
     Lemma~\ref{lem Ei+3-2} and thus is impossible. 
\end{enumerate}
Thus, \textbf{P9} is reducible.~$\blacklozenge$

   \medskip

  \noindent   \textbf{(P10)} Let $v$ be a $3^2$-vertex adjacent to $v'$ and 
  $1$-chain-adjacent to two $3^6$-vertices $v_1$ and $v_2$ with the internal vertices being $u_1$ and $u_2$ respectively.  
  Let the 
  $3$-chain-adjacent (resp., $2$-chain-adjacent) vertices of $v_1$ and $v_2$ be $v'_1$ and $v'_2$ (resp., $v''_1$ and $v''_2$), respectively. 
  Moreover, let the $3$-chains incident to $v_1$ (resp., $v_2$) 
  be $v_1u_{11}u_{12}u_{13}v'_1$ (resp., $v_2u_{21}u_{22}u_{23}v'_2$). 
  Applying Lemma~\ref{lem 3-3} on $v_1$ and $v_2$ we 
 get that $v''_1=v''_2=v'$. 
 Note that when we apply Lemma~\ref{lem 3-2} on $v_1$ (resp., $v_2$), 
 we know that the structures implied by 
point $(iv)$ of the lemma does not occur due to Lemma~\ref{lem Ei+3-2}. 
 
 Therefore, since the degree of $v$ is exactly $3$, 
applying  Lemma~\ref{lem 3-2} on $v_1$ and $v_2$ 
 forces $v'_1=v_2$ and $v'_2=v_1$. In this case, 
 either $v'$ is a cut vertex of $\overrightarrow{M}$, or 
 $\overrightarrow{M}$ is push equivalent to $\overrightarrow{E}_3$,  a contradiction. Thus, \textbf{P10} is reducible.~$\blacklozenge$

\medskip

This completes the proof of the lemma. 
 \end{proof}

\subsection{Discharging}
Let us define an initial charge function for the vertices of $\overrightarrow{M}$ as follows: 
$$ch(x) = 13deg(x)-30, \text{ for all } x \in V(\overrightarrow{M}).$$
Notice that 
$$\sum_{x \in V(\overrightarrow{M})} ch(x)  = \sum_{x \in V(\overrightarrow{M})} \left(13deg(x) - 30 \right) = 26|A(\overrightarrow{M})| - 30|V(\overrightarrow{M})| = - 2 \rho(\overrightarrow{M}) \leq 2.$$

\medskip

\noindent \textbf{Discharging rules:} We redistribute the charge of the vertices of $\overrightarrow{M}$ as per the following rules.

\begin{enumerate}[(R1)]
    \item A $3^+$-vertex $v$ donates a charge of $2$ to each of its chain-incident $2$-vertices $u$.

    \item A $3^+$-vertex $v$ donates a charge of $3$ to each of its adjacent $3^6$-vertices $u$.

    \item A $3^+$-vertex $v$ donates a charge of $1$ to each of its adjacent $3^5$-vertices $u$.

    \item A $3^+$-vertex $v$ donates a charge of $3$ to each of its 
    $1$-chain-adjacent $3^6$-vertices $u$.

    \item A $3^+$-vertex $v$ donates a charge of $1$ to each of its 
    $1$-chain-adjacent $3^5$-vertices $u$ when $v$ itself is not a 
    $3^5$-vertex. 
\end{enumerate}

Let $ch''(x)$ be the updated charge of the vertices of $\overrightarrow{M}$ after performing (R1)-(R5).

\begin{lemma}\label{lem updated charge}
    For each $x \in V(\overrightarrow{M})$, we have 
    $$ch''(x) \geq 
    \begin{cases}
     0 & \text{ if $x$ is a $2$-vertex},\\   
     3 & \text{ if $x$ is a $3^0$-vertex or a $3^1$-vertex},\\
     1 & \text{ if $x$ is a $3^2$-vertex or a 
                        $3^3$-vertex},\\
     0 & \text{ if $x$ is a $3^{\geq 4}$-vertex},\\
     3 & \text{ if $x$ is a $4^{\leq 9}$-vertex},\\     
     2 & \text{ if $x$ is a $4^{10}$-vertex},\\
     5 & \text{ if $x$ is a $5^{+}$-vertex},\\
    \end{cases}
    $$ 
\end{lemma}

\begin{proof}
If $x$ is a $2$-vertex, then only (R1) is applied on it. Since $x$ is chain-incident to exactly two $3^+$-vertex, and it receives a charge of $2$ from each of them, we have
$$ch''(x) = ch(x) + (2 \times 2) = (13 \times 2) - 30 + (2 \times 2) = 0.$$

If $x$ is a $3^0$-vertex or a $3^1$-vertex, then 
$x$ can have at most one chain-adjacent 
$3^6$-vertex since \textbf{P1} and \textbf{P3} are reducible configurations by Lemma~\ref{lem reducible-potential}. 
In particular, if $x$ is a $3^1$-vertex 
with one chain-adjacent $3^6$-vertex, then $x$ cannot have two other chain-adjacent $3^5$-vertices 
since \textbf{P7} is a reducible configuration by Lemma~\ref{lem reducible-potential}. Thus, in any case, $x$ donates a maximum of $6$ charge. Thus, 
$$ch''(x) \geq ch(x) - 6  = (13 \times 3) - 30 - 6 = 3.$$

If $x$ is a $3^2$-vertex,
then $x$ can have at most one 
chain-adjacent $3^6$-vertex 
at distance at most $2$.
Moreover, if $x$ has one $1$-chain-adjacent 
$3^6$-vertex, then it is not possible for $x$ to have two chain-adjacent $3^5$-vertices at distance at most $2$. 
The above inferences can be drawn as 
\textbf{P1}, \textbf{P4}, and \textbf{P10} are reducible configurations by Lemma~\ref{lem reducible-potential}. 
Therefore, the maximum charge $x$ donates is
$2\times 2 = 4$ to its chain-incident $2$-vertices, and either, $3$ to its chain-adjacent $3^6$-vertex at distance $2$ along with $1$ to at most one chain-adjacent $3^5$-vertex, or $1 \times 3 = 3$ to three chain-adjacent $3^5$-vertices. In any case, $x$ donates a maximum of $4+3+1 = 8$ charge. Thus, 
$$ch''(x) \geq ch(x) - 8  = (13 \times 3) - 30 - 8 = 1.$$

If $x$ is a $3^3$-vertex,
then $x$ does not have any 
chain-adjacent $3^6$-vertex 
at distance at most $2$.
Moreover,  $x$ 
cannot have three chain-adjacent $3^5$-vertices at distance at most $2$. 
The above inferences can be drawn as 
\textbf{P1}, \textbf{P6}, and \textbf{P9} are reducible configurations by Lemma~\ref{lem reducible-potential}. 
Therefore, the maximum charge $x$ donates is
$2\times 3 = 6$ to its chain-incident $2$-vertices, and $1 \times 2 = 2$ to its two 
chain-adjacent $3^5$-vertices at distance at most $2$. Thus, 
$$ch''(x) \geq ch(x) - 6 - 2  = (13 \times 3) - 30 - 8 = 1.$$

If $x$ is a $3^4$-vertex,
then $x$ does not have any 
chain-adjacent $3^6$-vertex 
at distance at most $2$, and can have at most one chain-adjacent $3^5$-vertex at distance at most $2$ since
\textbf{P1}, \textbf{P5}, \textbf{P8}, and \textbf{P9} are reducible configurations by Lemma~\ref{lem reducible-potential}. 
Therefore, the maximum charge $x$ donates is
$2\times 4 = 8$ to its chain-incident $2$-vertices, and $1$ to its  
chain-adjacent $3^5$-vertex at distance at most $2$. Thus, 
$$ch''(x) \geq ch(x) - 8 - 1  = (13 \times 3) - 30 - 9 = 0.$$

If $x$ is a $3^5$-vertex,
then $x$ does not have any 
chain-adjacent $3^6$-vertex 
at distance at most $2$, and does not have any 
adjacent $3^5$-vertex. In particular, if $x$ is a $3_{3,1,1}$-vertex, then it cannot have two $1$ chain-adjacent $3^5$-vertices and if $x$ is a $3_{2,2,1}$-vertex, then it cannot have any $1$ chain-adjacent $3^5$-vertex. That also means, $x$ must receive a charge of $1$ from some adjacent or $1$-chain-adjacent neighbor. 
These inferences can be drawn since 
\textbf{C4}, \textbf{P1}, \textbf{P2}, and \textbf{P9} are reducible by Lemmas~\ref{lem reducible-critical} and \ref{lem reducible-potential}.
Therefore, the maximum charge $x$ donates is
$2\times 5 = 10$ to its chain-incident $2$-vertices, and receives at least $1$ from its chain-adjacent vertices at distance at most $2$. Thus, 
$$ch''(x) \geq ch(x) - 10 + 1  = (13 \times 3) - 30 - 9 = 0.$$

If $x$ is a $3^6$-vertex,
then $x$ does not have any 
chain-adjacent $3^{\geq 5}$-vertex 
at distance at most $2$. Moreover, as $3_{2,2,2}$ is a reducible configuration, $x$ must have at least one chain-adjacent $3^+$-vertex at distance at most $2$. 
The above inferences can be drawn from Lemma~\ref{lem 3_222}, and as \textbf{P1} and \textbf{P9} are reducible configurations by Lemma~\ref{lem reducible-potential}. 
Therefore, the updated charge of $x$ is
$$ch''(x) \geq ch(x) - (2 \times 6) + 3  = (13 \times 3) - 30 - 12 +3 = 0.$$

If $x$ is a $4^{\leq 3}$-vertex,
then in the worst case scenario $x$ may have
four chain-adjacent $3^{6}$-vertex 
at distance at most $2$. 
Therefore, the updated charge of $x$ is
$$ch''(x) \geq ch(x) - (2 \times 3) - (3 \times 4)  = (13 \times 4) - 30 - 6 - 12 \geq 4.$$

If $x$ is a $4^{4}$-vertex,
then in the worst case scenario $x$ may have
three chain-adjacent $3^{6}$-vertices and one $3^5$ vertex at distance at most $2$
 since
\textbf{C13} is a reducible configuration by Lemma~\ref{lem reducible-critical}.
Therefore, the updated charge of $x$ is
$$ch''(x) \geq ch(x) - (2 \times 4) - (3 \times 3) - 1  = (13 \times 4) - 30 - 8 - 9 - 1 \geq 4.$$

If $x$ is a $4^{5}$-vertex,
then in the worst case scenario $x$ may have
three chain-adjacent $3^{\geq 5}$-vertices at distance at most $2$.
Therefore, the updated charge of $x$ is
$$ch''(x) \geq ch(x) - (2 \times 5) - (3 \times 3)  = (13 \times 4) - 30 - 10 - 9 \geq 3.$$

If $x$ is a $4^{6}$-vertex,
then in the worst case scenario $x$ may have
three chain-adjacent $3^{\geq 5}$-vertices at distance at most $2$ among which not all can be $3^6$-vertices
 since \textbf{C12} is a  reducible configuration by Lemma~\ref{lem reducible-critical}.
Therefore, the updated charge of $x$ is
$$ch''(x) \geq ch(x) - (2 \times 6) - (3 \times 2) - 1  = (13 \times 4) - 30 - 12 - 6 -1 \geq 3.$$

If $x$ is a $4^{7}$-vertex,
then in the worst case scenario $x$ may have
two chain-adjacent $3^{\geq 5}$-vertices at distance at most $2$ among which not all can be $3^6$-vertices
 since
\textbf{C9} and \textbf{C11} are reducible configurations by Lemma~\ref{lem reducible-critical}.
Therefore, the updated charge of $x$ is
$$ch''(x) \geq ch(x) - (2 \times 7) - 3  - 1  = (13 \times 4) - 30 - 14 - 3 -1 \geq 4.$$

If $x$ is a $4^{8}$-vertex,
then in the worst case scenario $x$ may have
one chain-adjacent $3^{\geq 5}$-vertex at distance at most $2$ 
 since \textbf{C11} is a reducible configuration by Lemma~\ref{lem reducible-critical}.
Therefore, the updated charge of $x$ is
$$ch''(x) \geq ch(x) - (2 \times 8) - 3  = (13 \times 4) - 30 - 16 - 3 \geq 3.$$

If $x$ is a $4^{9}$-vertex,
then in the worst case scenario $x$ may have
one chain-adjacent $3^{\geq 5}$-vertex at distance at most $2$, which, in fact, cannot be a $3^6$ vertex,
 since \textbf{C5} and \textbf{C6} are reducible configurations by Lemma~\ref{lem reducible-critical}.
Therefore, the updated charge of $x$ is
$$ch''(x) \geq ch(x) - (2 \times 9) - 1  = (13 \times 4) - 30 - 18 - 1 \geq 3.$$

If $x$ is a $4^{10}$-vertex,
then  $x$ does not have
any chain-adjacent $3^{\geq 5}$-vertex at distance at most $2$
 since \textbf{C7} is a reducible configuration by Lemma~\ref{lem reducible-critical}.
Therefore, the updated charge of $x$ is
$$ch''(x) \geq ch(x) - (2 \times 10)  = (13 \times 4) - 30 - 20 \geq 2.$$

If $x$ is a $k$-vertex for some $k \geq 5$,
then in the worst case scenario $x$ has 
$3k$ many chain-incident $2$-vertices
 since \textbf{C1} is reducible configurations by Lemma~\ref{lem reducible-critical}.
Therefore, the update charge of $x$ is
$$ch''(x) \geq ch(x) - (2 \times 3k)  = (13 \times k) - 30 -6k = 7k - 30 \geq 5.$$

This completes the proof. 
\end{proof}

Note that we have already shown that the updated charge of each vertex is non-negative, that is, $ch''(x) \geq 0$ for all $x \in V(\overrightarrow{M})$. However, since $\sum_{x \in V(\overrightarrow{M})} ch(x) \leq 2$, it is, thus, not possible for a set of vertices of $\overrightarrow{M}$ to have updated charge $3$ or more collectively. Using this, we are going arrive at a contradiction by showing that 
after performing the discharging according to (R1)-(R5), the updated total charge is
   $\sum_{x \in V(\overrightarrow{M})} ch''(x) \geq 3$. 
   However, the proof is long, and thus, for convenience of the readers, we divide the proof into several lemmas. 
\begin{lemma}\label{lem M is subcubic}
    The oriented graph $\overrightarrow{M}$ is subcubic. 
\end{lemma}
\begin{proof}
 From Lemma~\ref{lem updated charge}, it is clear that $ch''(x) \geq 3$ if $x$ is a $5^+$-vertex or a 
 $4^{\leq 9}$-vertex, and $ch''(x) \geq 2$ if $x$ is a $4^{10}$-vertex. This means that if $\overrightarrow{M}$ is not subcubic, then it can have exactly one $4^+$ vertex $u$, and $u$ must necessarily be a $4^{10}$-vertex. Furthermore, we know that $ch''(x) \geq 1$ if $x$ is a $3^{\leq 3}$-vertex due to Lemma~\ref{lem updated charge}. Thus, apart from $u$, all vertices of $\overrightarrow{M}$ must be $3^k$-vertices, 
 where $k \in \{4,5,6\}$. 

\medskip 

\noindent \textbf{Fact 1:} A $3^4$-vertex $x$ has updated charge $1$ unless it has a chain-adjacent $3^5$-vertex at distance at most $2$. 

\medskip 

\noindent \textbf{Fact 2:} A $3^5$-vertex $x$ has updated charge $1$ if it has two chain-adjacent $3^{\leq 4}$-vertices at distance at most $2$. 

\medskip

Thus, we can assume that any $3^4$-vertex has exactly one chain-adjacent $3^5$-vertex at distance at most $2$ and any $3^5$-vertex has exactly one 
 chain adjacent $3^4$-vertex at distance at most $2$.

\medskip 

\noindent \textbf{Fact 3:} If $x$ is a $3^6$ vertex of $\overrightarrow{M}$, then it must be a $3_{3,2,1}$-vertex. 

\medskip

\noindent \textit{Proof of the fact.}  If $x$ is a $3^6$-vertex, then it cannot be a $3_{2,2,2}$-vertex due to 
        Lemma~\ref{lem 3_222}. Suppose that $x$ is a $3_{3,3,0}$-vertex adjacent to $y$. If $y$ is a $3$-vertex, then 
        $y$ has to be a $3^0$-vertex as \textbf{P1} is reducible due to Lemma~\ref{lem reducible-potential}. 
        Note that, 
         every $4^{10}$-vertex is adjacent to four distinct $2$-vertices 
        as \textbf{C1} is reducible due to 
        Lemma~\ref{lem reducible-critical}. So if $y$ is a $4^+$-vertex, it cannot be a $4^{10}$-vertex. Thus, by Lemma~\ref{lem updated charge} we can conclude that $ch''(y) \geq 3$. That means, the existence of a $3_{3,3,0}$-vertex $x$ 
        implies the existence of a vertex $y$ having $ch''(y) \geq 3$ leading to a contradiction. 
        So $x$ cannot be a $3_{3,3,0}$-vertex, and hence must be a $3_{3,2,1}$-vertex.~$\blacksquare$

\medskip  

\noindent \textbf{Claim 1:} $\overrightarrow{M}$ does not contain any $3^6$-vertex or any $3_{2,2,1}$-vertex. 

\medskip

\noindent \textit{Proof of the claim.} Let $x$ be a $3_{3,2,1}$-vertex or a $3_{2,2,1}$-vertex $1$-chain-adjacent to $y$. Notice that $y \neq u$ as \textbf{C7} is reducible due to Lemma~\ref{lem reducible-critical}. 
Moreover, $y$ must be a $3^{\leq 3}$-vertex due to Lemma~\ref{lm 3_221-1-3>4} and since \textbf{P9} is reducible due to Lemma~\ref{lem reducible-potential}. 
However, we have established earlier that $\overrightarrow{M}$ does not contain a $3^{\leq 3}$-vertex, which is a contradiction. Thus, we can conclude that $\overrightarrow{M}$ does not 
contain a $3_{3,2,1}$-vertex or a $3_{2,2,1}$-vertex. ~$\blacksquare$

\medskip  

\noindent \textbf{Claim 2:} $\overrightarrow{M}$ does not contain any  $3_{3,2,0}$-vertex. 

\medskip

\noindent \textit{Proof of the claim.} Let us assume that $x$ is a $3_{3,2,0}$-vertex adjacent to $y$. 
Observe that $y$ must be a $3^4$-vertex since the configuration \textbf{C4} is reducible from Lemma~\ref{lem reducible-critical}.
Note that, applying Lemma~\ref{lem 3-3} on $x$, without loss of generality,
will force a directed $6$-cycle $\overrightarrow{C}$ in 
$\overrightarrow{M}$ of the form $xyu_1zu_2u_3x$ where $u_1, u_2$ and $u_3$ 
are $2$-vertices. This, in particular, implies that $y$ is a 
$3_{3,1,0}$-vertex. 
In that case, $z$ must be a $3^4$-vertex. 
This is exactly the configuration \textbf{C15} which is reducible due to Lemma~\ref{lem reducible-critical}. Therefore, we can conclude that $\overrightarrow{M}$ does not contain a $3_{3,2,0}$-vertex.~$\blacksquare$

\medskip  

\noindent \textbf{Claim 3:} $\overrightarrow{M}$ does not contain any 
$3_{3,1,1}$-vertex. 

\medskip

\noindent \textit{Proof of the claim.} Let $x$ be a $3_{3,1,1}$-vertex $1$-chain adjacent to $y$ and $z$. Applying Lemma~\ref{lem 3-3} on $x$, without loss of generality,
will force a directed $6$-cycle $\overrightarrow{C}$ in 
$\overrightarrow{M}$ of the form $xu_1yu_2zu_3x$ where $u_1$ and $u_3$ 
are $2$-vertices. If $u_2$ is not a $2$-vertex, then it must be a $3^{\leq 3}$-vertex since \textbf{C1} is reducible due to Lemma~\ref{lem reducible-critical}, which is not possible. 
Thus, $u_2$ must also be a $2$-vertex. 
Since every $3^5$-vertex must have exactly one chain-adjacent 
$3^4$-vertex at distance at most $2$, without loss of generality, we may 
assume that $y$ is a $3^4$-vertex and $z$ is a $3^{\geq 4}$-vertex. This is impossible as \textbf{C16} is reducible by Lemmas~\ref{lem reducible-critical}. This means, there are no $3^5$-vertices in the graph $\overrightarrow{M}$.~$\blacksquare$

\medskip  

Therefore, the only $3^+$-vertices of $\overrightarrow{M}$ other than $u$ are $3^4$-vertices. However, if a $3^4$-vertex $x$ is neither adjacent nor $1$-adjacent to any $3^5$-vertex, then $ch''(x) \geq 1$, a contradiction. 
So $\overrightarrow{M}$ cannot have a $3^4$-vertex as well. That means, the only $3^+$ vertex in $\overrightarrow{M}$ is the $4^{10}$-vertex $u$. However, this is not possible since the existence of a $4^{10}$-vertex $u$ forces the existence of  four other $3^+$-vertices chain-adjacent to $u$. 
Thus, we can conclude that $\overrightarrow{M}$ must be a subcubic graph. 
\end{proof}

\begin{lemma}\label{lem no 3^6}
    The oriented graph $\overrightarrow{M}$ does not contain any $3^6$-vertex.  
\end{lemma}

\begin{proof}
Suppose $\overrightarrow{M}$ contain a $3^6$-vertex $x$. We know   
that $x$ must be a $3_{3,2,1}$-vertex due to Fact 3 of Lemma~\ref{lem M is subcubic}. 
Assume that $x$ is $1$-chain-adjacent to  $y$ and 
$2$-chain-adjacent to $z$.  
Applying Lemma~\ref{lem 3-3} on $x$, without loss of generality,
will force a directed $6$-cycle $\overrightarrow{C}$ in 
$\overrightarrow{M}$ of the form $xu_1yzu_2u_3x$ where $u_1, u_2, u_3$ 
are $2$-vertices. 
Notice that, $y$ must be a $3^{2}$-vertex 
as \textbf{P9} is reducible due to Lemma~\ref{lem reducible-potential}, and by using Fact 1. 
Also $z$ must be a $3^{\leq 4}$-vertex as \textbf{P4} is reducible 
due to Lemma~\ref{lem reducible-potential}. If $z$ is a $3^4$-vertex then it is neither adjacent nor $1$-chain-adjacent to any $3^5$-vertex. 
Hence $ch''(y) \geq 1$ and $ch''(z) \geq 1$. Thus, $\overrightarrow{M}$ cannot contain any vertex other than $y$ and $z$ with updated charge at least $1$. 

Therefore, apart from $x$ and (may be) $y$, all vertices of $\overrightarrow{M}$ must be $3^k$-vertices, where $k \in \{4,5\}$. 
Moreover, we can assume that any $3^4$-vertex (except, may be, $z$) has exactly one chain-adjacent $3^5$-vertex at distance at most $2$ and any $3^5$-vertex has exactly one 
 chain adjacent $3^4$-vertex at distance at most $2$.

 \medskip

 \noindent \textbf{Claim 4:} $\overrightarrow{M}$ does not contain any $3^5$-vertex. 

 \medskip

 \noindent \textit{Proof of the claim.}  Using the exact same arguments used to prove Claim 1 (resp., Claim 2, Claim 3) from the proof of Lemma~\ref{lem M is subcubic}, one can show that 
 $\overrightarrow{M}$ does not contain any $3_{2,2,1}$-vertex (resp., $3_{3,2,0}$-vertex, $3_{3,1,1}$-vertex).~$\blacksquare$

\medskip

Suppose that $y$ is a $3^2$-vertex chain-adjacent to a $3^6$-vertex $x$, a $3^{\leq 4}$-vertex $z$, and another $3$-vertex $w$. 
We know that $w$ cannot be a $3^{\geq 5}$-vertex. Therefore, the updated charge of $y$ must be at least $2$, that is, $ch''(y) \geq 2$. This is a contradiction. 

Thus, we can conclude that $\overrightarrow{M}$ 
does not contain any $3^6$-vertex. 
\end{proof}

\begin{lemma}\label{lem no 3-311}
    The oriented graph $\overrightarrow{M}$ does not contain any $3_{3,1,1}$-vertex.  
\end{lemma}

\begin{proof}
Suppose $\overrightarrow{M}$ contains a $3_{3,1,1}$-vertex $x$ 
$1$-chain-adjacent to $y$ and $z$. 
Applying Lemma~\ref{lem 3-3} on $x$, without loss of generality,
will force a directed $6$-cycle $\overrightarrow{C}$ in 
$\overrightarrow{M}$ of the form $xu_1yu_2zu_3x$ where $u_1$ and $u_3$ 
are $2$-vertices. 
If $u_2$ is a $3^+$-vertex, then 
$y,z$ must be $3^{\leq 4}$-vertices and $u_2$ must be a 
$3^{\leq 3}$-vertex since \textbf{C1} is reducible due to 
Lemma~\ref{lem reducible-critical}. 
In particular, $u_2$ is a $3^{\leq 3}$-vertex 
having no chain-adjacent 
$3^5$-vertex at distance at most $2$. Therefore, we must have $ch''(u_2) \geq 3$, a contradiction. Thus, $u_2$ must be a $2$-vertex.

Since \textbf{C16} is reducible due to Lemma~\ref{lem reducible-critical}, without loss of generality we may assume that $y$ is a $3^{\leq 3}$-vertex and $z$ is a $3^2$-vertex. Note that, $z$ can have at most two 
chain-adjacent 
$3^5$-vertex at distance at most $2$. Thus, $ch''(z) \geq 3$, a contradiction. 
Thus, we can conclude that $\overrightarrow{M}$ 
does not contain any $3_{3,1,1}$-vertex. 
\end{proof}

\begin{lemma}\label{lem no 3-320}
    The oriented graph $\overrightarrow{M}$ does not contain any $3_{3,2,0}$-vertex.  
\end{lemma}

\begin{proof}
Suppose $\overrightarrow{M}$ contains a $3_{3,2,0}$-vertex $x$ 
adjacent to $y$ and $2$-chain-adjacent to  $z$. 
Applying Lemma~\ref{lem 3-3} on $x$, without loss of generality,
will force a directed $6$-cycle $\overrightarrow{C}$ in 
$\overrightarrow{M}$ of the form $xyu_1zu_2u_3x$ where $u_1, u_2$ and $u_3$ 
are $2$-vertices. 

Since \textbf{C14} and \textbf{C15} are reducible due to Lemma~\ref{lem reducible-critical} and $\overrightarrow{M}$ does not contain any $3^{\leq 1}$-vertex, without loss of generality we may assume that $y$ is a $3^{\leq 3}$-vertex and $z$ is a $3^3$-vertex. 
Note that,  $z$ can have at most  one  
chain-adjacent 
$3^5$-vertex at distance at most $2$. Thus, $ch''(y) \geq 1$ (by Lemma~\ref{lem updated charge}), and $ch''(z) \geq 2$, a contradiction. 
Thus, we can conclude that $\overrightarrow{M}$ 
does not contain any $3_{3,2,0}$-vertex. 
\end{proof}

\begin{lemma}\label{lem no 3-2}
    The oriented graph $\overrightarrow{M}$ does not contain any 
    $3^2$-vertex.  
\end{lemma}

\begin{proof}
Suppose $\overrightarrow{M}$ contains a $3^2$-vertex $x$ 
adjacent to $y$. Note that, since the only type of $3^5$-vertex that can be present in $\overrightarrow{M}$ is a $3_{2,2,1}$-vertex due to Lemma~\ref{lem no 3-311} and~\ref{lem no 3-320}, $y$ must be a $3^{\leq 4}$-vertex. Hence, $x$ can have at most two chain adjacent $3^5$-vertices. Thus, $ch''(x) \geq 3$, a contradiction. 
Hence, we can conclude that $\overrightarrow{M}$ 
does not contain any $3^2$-vertex. 
\end{proof}

\bigskip

\noindent \textit{Proof of Theorem~\ref{th push 3-critical density}.} 
Note that any  $3^5$-vertex in $\overrightarrow{M}$ must be a 
$3_{2,2,1}$-vertex due to 
Lemmas~\ref{lem no 3-311} and~\ref{lem no 3-320}. 
Also, every $3^5$-vertex must be $1$-chain-adjacent to a $3^3$-vertex 
since it cannot be a $3^{\geq 4}$-vertex due to Lemma~\ref{lm 3_221-1-3>4} and it cannot be 
a $3^{\leq 2}$ vertex as we have already established non-existence of 
such vertices in $\overrightarrow{M}$.
Moreover, observe that any 
$3^4$-vertex in $\overrightarrow{M}$ now has updated charge at least $1$. That means, every $3^{\leq 4}$-vertex has updated charge at least $1$. 

A $3_{3,1,0}$-vertex (or a $3_{2,1,1}$-vertex) $x$ must have two distinct chain-adjacent $3^{\leq 4}$-vertices $y$ and $z$ at distance at most $2$. This will imply 
$ch''(x) + ch''(y) + ch''(z) \geq 3$, a contradiction. Therefore, $\overrightarrow{M}$ does not have any $3_{3,1,0}$-vertex or any $3_{2,1,1}$-vertex.

A $3_{3,0,0}$-vertex $x$ does not have any chain-adjacent 
$3^5$-vertex at distance at most $2$, 
and thus has $ch''(x) \geq 3$, a contradiction. 
On the other hand, a $3_{2,1,0}$-vertex $x$ has at most one chain-adjacent $3^5$-vertex, and must have one adjacent 
$3^{\leq 4}$-vertex $y$. Notice that, $ch''(x) \geq 2$ and $ch''(y) \geq 1$, a contradiction. That means, 
$\overrightarrow{M}$ does not have any $3_{3,0,0}$-vertex and any $3_{2,1,0}$-vertex.

So the only types of $3^{\leq 4}$-vertices that $\overrightarrow{M}$ can have are $3_{2,2,0}$-vertices and $3_{1,1,1}$-vertices. 
Moreover, the only type of $3^{\geq 5}$-vertices $\overrightarrow{M}$ can have are $3_{2,2,1}$-vertices.

Suppose that $x$ is a $3_{2,2,0}$-vertex. 
Notice that it must be adjacent to another $3_{2,2,0}$-vertex $y$. 
If $x$ is $2$-chain-adjacent to a  $3^{\leq 4}$-vertex $z$, then we will have 
$ch''(x) + ch''(y) + ch''(z) \geq 3$, a contradiction. 
If $x$ is $2$-chain-adjacent to a $3_{2,2,1}$-vertex, which must be $1$-chain-adjacent to a $3^3$-vertex $z$ 
(distinct from $x,y$), then also we will have 
$ch''(x) + ch''(y) + ch''(z) \geq 3$, a contradiction. Therefore, $\overrightarrow{M}$ does not have any $3^4$-vertices at all.

Suppose that $x$ is a $3_{1,1,1}$-vertex. It must have a $1$-chain-adjacent $3_{1,1,1}$-vertex $y$ since 
 \textbf{P6} is reducible by 
 Lemma~\ref{lem reducible-potential}.  If $x$ (or $y$) has another $1$-chain-adjacent $3^{\leq 4}$-neighbor $z$, then 
 we will have 
$ch''(x) + ch''(y) + ch''(z) \geq 3$, a contradiction. 
If $x$ is $2$-chain-adjacent to a $3_{2,2,1}$-vertex, which must be $1$-chain-adjacent to a $3^3$-vertex $z$ 
(distinct from $x,y$), then also we will have 
$ch''(x) + ch''(y) + ch''(z) \geq 3$, a contradiction. 
Therefore, $\overrightarrow{M}$ does not have any $3^3$-vertices at all. This also implies that the only type of vertices in $\overrightarrow{M}$ is $3_{2,2,1}$-vertex. However, this is impossible by Lemma~\ref{lm 3_221-1-3>4}. Thus, 
$\overrightarrow{M}$ cannot exist. 
 \qed

\bibliographystyle{abbrv}
\bibliography{reference.bib}

\end{document}